\documentclass{article}
\usepackage[nonatbib,preprint]{neurips_2020}
\usepackage[utf8]{inputenc} % allow utf-8 input
\usepackage[T1]{fontenc}    % use 8-bit T1 fonts
\usepackage{hyperref}       % hyperlinks
\usepackage{url}            % simple URL typesetting
\usepackage{booktabs}       % professional-quality tables
\usepackage{amsfonts}       % blackboard math symbols
\usepackage{nicefrac}       % compact symbols for 1/2, etc.
\usepackage{microtype}      % microtypography

\usepackage[linesnumbered,ruled]{algorithm2e}

\let\oldnl\nl% Store \nl in \oldnl
\newcommand{\nonl}{\renewcommand{\nl}{\let\nl\oldnl}} 

\usepackage{enumitem}
\usepackage{amsmath,amsfonts,mathrsfs,amssymb,bm,mathtools,mathabx}
\usepackage{graphicx, epsfig, epstopdf, mdframed}
\usepackage{wrapfig, multirow, multicol}
\usepackage{url}

\usepackage{float} % FORBIDDEN PACKAGE
\floatstyle{ruled}
\newfloat{model}{thp}{lop}
\floatname{model}{Model}

\usepackage{subcaption}
\usepackage{cleveref}

\newcommand{\bb}[1]{\textbf{#1}}
\newcommand*{\defeq}{\stackrel{\text{def}}{=}}

\newcommand{\userset}{\cP}
\newcommand{\unitset}{\cU}
\newcommand{\regionset}{\cR}
\newcommand{\unitsizeset}{\cZ}
\newcommand{\sizeset}{\cS}
\newcommand{\user}{p}
\newcommand{\unit}{u}
\newcommand{\region}{r}
\newcommand{\unitsize}{z}
\newcommand{\size}{\sigma}
\newcommand{\group}{G}
\newcommand{\gsize}{n}
\newcommand{\gsizevec}{\bm{n}}
\newcommand{\csize}{c}
\newcommand{\csizevec}{\bm{c}}

\newcommand{\node}{a}
\newcommand{\US}{\textsl{US}}
\newcommand{\GA}{\textsl{GA}}
\newcommand{\MI}{\textsl{NY}}
\newcommand{\ctable}{\bm{\tau}}
\newcommand{\cfunction}{\mathring{\ctable}}
\newcommand{\ctablech}{\bm{\phi}}
\newcommand{\cfunctionch}{\mathring{\ctablech}}

\newcommand{\dom}{D}
\newcommand{\tree}{\cT}
\newcommand{\treedp}{\cT^{\textsl{dp}}}
\newcommand{\chaindp}{\cT^{\textsl{ch}}}
\def\thmspace{0.2em}
\newtheorem{theorem}{\hspace{\thmspace}{\bf Theorem}\!}

\newtheorem{definition}{\hspace{\thmspace}{\bf Definition}\!}

\newtheorem{lemma}{\hspace{\thmspace}{\bf Lemma}\!}
\newenvironment{proof}{{\textit{Proof}.}}{\hfill$\Box$}
\newtheorem{example}{\hspace{\thmspace}{\bf Example}\!}

\def\st{\: | \:}

\newcommand{\argmin}{\operatornamewithlimits{argmin}}
\newcommand{\minimize}{\operatornamewithlimits{Minimize~}}
\newcommand{\subjectto}{\operatornamewithlimits{Subject~to:}}

\newcommand{\setf}[1]{{\bf{#1}}}

 \newcommand{\cM}{\mathcal{M}}
 
 \newcommand{\cR}{\mathcal{R}}
 \newcommand{\cS}{\mathcal{S}}
 \newcommand{\cT}{\mathcal{T}}
 \newcommand{\cU}{\mathcal{U}}
 \newcommand{\cP}{\mathcal{P}}
 
\newcommand{\cZ}{\mathcal{Z}}

\newcommand{\bv}{\setf{v}}
\newcommand{\sR}{{\mathscr{R}}}
\newcommand{\sD}{{\mathscr{D}}}	% the set of all databases

\newcommand{\RR}{\mathbb{R}}
\newcommand{\NN}{\mathbb{N}}

\title{Differential Privacy of Hierarchical Census Data: 
	\\An Optimization Approach}

\author{%
  Ferdinando Fioretto \\
  {\small Syracuse University} \\
  \texttt{ffiorett@syr.edu}\\
  \And
  Pascal Van Hentenryck\\
  {\small Georgia Institute of Technology}\\
  \texttt{pvh@isye.gatech.edu}
  \And
  Keyu Zhu\\
  {\small Georgia Institute of Technology}\\
  \texttt{{kzhu67@gatech.edu}}
}

\begin{document}
\maketitle\sloppy\allowdisplaybreaks
% %\allowdisplaybreaks
% \begin{frontmatter}
% \title{Differential Privacy of Hierarchical Census Data: 
% \\An Optimization Approach}

% \author{Ferdinando Fioretto\corref{cor1}\fnref{auth1}}
% \ead{ffiorett@syr.edu}
% \author{Pascal Van Hentenryck\fnref{auth2}}
% \ead{pvh@isye.gatech.edu}
% \author{Keyu Zhu\fnref{auth2}}
% \ead{kzhu67@gatech.edu}

% \address[auth1]{Syracuse University, Syracuse, New York, USA}
% \address[auth2]{Georgia Institute of Technology, Atlanta, Georgia, USA}
% \cortext[cor1]{corresponding author}
%--------------------------------------------------------------------

\begin{abstract}
This paper is motivated by applications of a Census Bureau interested
in releasing aggregate socio-economic data about a large population
without revealing sensitive information about any individual. The
released information can be the number of individuals living alone,
the number of cars they own, or their salary brackets. Recent events
have identified some of the privacy challenges faced by these
organizations \cite{abowd2018us}. To address them, this paper presents
a novel differential-privacy mechanism for releasing hierarchical
counts of individuals. The counts are reported at multiple
granularities (e.g., the national, state, and county levels) and must
be consistent across all levels. The core of the mechanism is an
optimization model that redistributes the noise introduced to achieve
differential privacy in order to meet the consistency constraints
between the hierarchical levels.  The key technical contribution of
the paper shows that \emph{this optimization problem can be solved in
polynomial time by exploiting the structure   of its cost functions}.
Experimental results on very large, real datasets show that the
proposed mechanism provides improvements of up to two orders of
magnitude in terms of computational efficiency and accuracy with
respect to other state-of-the-art techniques. \end{abstract}

% \begin{keyword}
% Differential Privacy, Constrained Optimization, Census
% \end{keyword}

% \end{frontmatter}
% \linenumbers

%%%%%%%%%%%%%%%%%%%%%%%%%%%%%%%%%%%
\section{Introduction}
\label{sec:Introduction}
%%%%%%%%%%%%%%%%%%%%%%%%%%%%%%%%%%%

The release of datasets containing sensitive information about a 
large number of individuals is central to a number of statistical
analysis and machine learning tasks.  For instance, the US Census
Bureau publishes socio-economic information about individuals, which
is then used as input to train classifiers or predictors, release
important statistics about the US population, and take decisions
relative to elections and financial aid.  %% One of the fundamental
roles of a Census Bureau is to report \emph{group size} queries, which
are especially useful to study the skewness of a distribution. For
instance, in 2010, the US Census Bureau released 33 datasets of such
queries \cite{census-files}. Group size queries partition a dataset in
\emph{groups} and evaluate the size of each group. For instance, a
group may be the households that are families of four members, or the
households owning three  cars. 

The challenge is to release these datasets without disclosing
sensitive information about any individual in the dataset. The
confidentiality of information in the decennial census is also
required by law.  Various techniques for limiting a-priori the
disclosed information have been investigated in the past, including
anonymization \cite{Sweeney:02} and aggregations \cite{winkler:02}. 
However, these techniques have been consistently shown ineffective in
protecting sensitive data \cite{Golle:06,Sweeney:02}, For instance,
the US Census Bureau confirmed \cite{census-leak} that the disclosure
limitations used for the 2000 and 2010 censuses had serious
vulnerabilities which were exposed by the Dinur and Nissim's
reconstruction attack \cite{dinur:03}. Additionally, the 2010 Census
group sizes were truncated due to the lack of privacy methods for
protecting these particular groups \cite{censu-attack}. 

This paper addresses these limitations through the framework of \emph
{Differential Privacy} \cite{Dwork:06}, a formal approach to guarantee
data privacy by bounding the disclosure risk of any  individual
participating in a dataset. Differential privacy is  considered the
de-facto standard for privacy protection and has been  adopted by
various corporations \cite{erlingsson2014rappor,apple}  and
governmental agencies \cite{census-DP}.  Importantly, the 2020 US
Census will use an approach to disclosure avoidance that satisfies the
notion of Differential Privacy \cite{abowd2018us}.

Differential privacy works by injecting carefully calibrated noise to
the data before release. However, whereas this process guarantees
privacy, it also affects the fidelity  of the released data. In
particular, the injected noise often  produces datasets that violate
consistency constraints of the  application domain. In particular,
group size queries must be  consistent in a geographical hierarchy,
e.g., the national, state,  and county levels. Unfortunately, the
traditional injection of  independent noise to the group sizes cannot
ensure the consistency  of hierarchical constraints.

To overcome this limitation, this paper casts the problem of 
privately releasing group size data as a \emph{constraint 
optimization problem} that ensures consistency of the hierarchical 
dependencies. However, the optimization problem that redistributes 
noise optimally is intractable for real datasets involving hundreds 
of millions of individuals. In fact, even its convex relaxation, 
which does not guarantee consistency, is challenging  computationally.
This paper addresses these challenges by proposing  mechanisms based
on a dynamic programming scheme that leverages both  the hierarchical
nature of the problem and the structure of the  objective function.  

\begin{table}[!tb]
\centering
%\resizebox{\linewidth}{!} 
 {%%
 \begin{tabular}{@{}rl@{}}
 \toprule
 Notation & Description \\
 \midrule
 	$\userset$ & The set of all users \\
 	$\unitset$ & The set of all units (e.g., home addresses)\\ 
 	$\regionset$ & The set of all regions (e.g., census blocks, states)\\
 	$\unitsizeset$ & The set of all unit quantities (e.g., the number of cars a user owns) \\
	$\sizeset$ & The set of all group sizes \\
 	% $\user_i$ & A user in $\userset$ (e.g., an ID)\\
 	% $\unit_i$ & A unit in $\unitset$ (e.g., an home address)\\
 	% $\region_i$ & A region in $\regionset$ (e.g., a census block)\\
 	% $\unitsize_i$ & A unit quantity (e.g., the number of cars the user $p_i$ owns)\\ 
	$\cT$ & The region hierarchy \\
	$\regionset_\ell$ & The set of regions in level $\ell$ of $\cT$ \\
	$L$ & The number of levels in $\cT$ \\
	$N$ & The number of group sizes, i.e., $|\sizeset|$\\
	$G$ & The total count of groups, i.e., 
	$ \sum_{r \in R_\ell} \sum_{i=1}^N \gsize^r_i $ for any $\ell \in [L]$\\
	$n$  & The number of individuals in the dataset \\
	$\gsize_s^r$ & The number of groups of size $s$ in region $r$ \\
	$\gsizevec^r$ & The vector of group sizes for region $r$ \\
	$\node_s^r$ & A node in the DP tree associated to region $r$ and group size $s$ \\
	$\ctable_s^r$ & The cost table associated to node $\node_s^r$\\
	$\ctable_s^r(v)$ & The cost value of $\ctable_s^r$ associated to value $v$\\
	$\ctablech_s^r$ & The contribution of $a^r_s$ children costs tables\\
	$\dom_s^r$ & The domain associated to node $\node_s^r$\\
	$ch(r)$ & The set of children regions of region $r$ in $\cT$ \\
	$pa(r)$ & The parent region of region $r$ in $\cT$ \\
  \bottomrule
 \end{tabular}
 }
 \caption{Important symbols adopted in the paper.}
 \label{tab:symbols}
\end{table}
\paragraph{Paper Contributions}
The paper focuses on the \emph{Privacy-preserving Group Size Release}
(PGSR) problem for releasing differentially private group sizes that
preserves hierarchical consistency.
Its contributions are summarized  as follows: 
{\bf(1)}
It proposes several differentially private mechanisms that rely on an 
optimization approach to release both accurate and consistent group sizes. 
{\bf(2)}
It shows that the differentially private mechanisms can be implemented
in polynomial time, using a dynamic program that exploits the
hierarchical nature of group size queries, the structure of the
objective functions, and cumulative counts. 
{\bf(3)} Finally, it
evaluates the mechanisms on very large datasets containing over
300,000,000 individuals. The results demonstrate the effectiveness and
scalability of the proposed mechanisms that bring several orders of
magnitude improvements over the state of the art. 
%This paper extends an early work presented in~\cite{fioretto:CP-19}. 
%: it adds a new
% \emph{exact} and polynomial-time dynamic programming mechanism for the
% cumulative count problem that overcomes a limitation faced by the
% \emph{approximate} solution proposed in the conference version of this
% paper. This novel mechanism almost always dominates all the others in
% terms of accuracy and runs signiificantly faster. This paper also
% provides the proofs for all theorems and lemmas, which were omitted in
% \cite{fioretto:CP-19}, as well as  extended experimental evaluation of
% the proposed mechanisms.

\paragraph{Paper Organization} The paper is organized as follows.
\Cref{sec:ProblemDefinition} introduces the notation and describes the
group size release problem.   \Cref{sec:differential_privacy} reviews
the privacy notion adopted in this work, as well as some useful
results  adopted in the privacy analysis of the proposed mechanisms.
\Cref{sec:group_est} presents the Privacy-Preserving Group Size
Release (PGSR) problem  and discusses its privacy, consistency,
validity, and faithfulness criteria. \Cref{sec:IPPGSR} presents
$\cM_H$, a two-step mechanism for the PGSR problem that uses an
optimization-based post-processing step to satisfy the PGSR criteria. 
\Cref{sec:dp_gse} presents $\cM_H^{\textsl{dp}}$, an exact mechanism
for the PGSR problem that exploits dynamic programming. While the
proposed mechanism is exact, it becomes  ineffective for very large
real-life applications. An efficient polynomial-time solution for
solving the dynamic program is presented in \Cref{sec:dp_gse_exploit}.
 \Cref{sec:cumulative_mechanisms} discusses how to use
\emph{cumulative queries} to reduce the amount of noise required to
produce the privacy-preserving counts and \Cref{sec:cumulative_PGSR}
presents $\cM_c$, a \emph{sub-optimal} algorithm to solve the PGSR
under cumulative queries.  Then, \Cref{sec:_apprdp_cumulative_PGSR}
and \Cref{sec:dp_cumulative_PGSR} present, respectively, an
approximate and exact efficient dynamic-programming mechanism for
solving the PGSR problem under cumulative queries.  Finally,
\Cref{sec:related_work} discusses related work, \Cref{sec:experiments}
report an evaluation of the proposed mechanisms on several realistic
datasets, and \Cref{sec:conclusions} concludes the work. 

A summary of the important symbols adopted in the paper is provided in Table
\ref{tab:symbols}.  The appendix contains the proofs for all lemmas
and theorems.

%%%%%%%%%%%%%%%%%%%%%%%%%%%%%%%%%%%
\section{Problem Specification}
\label{sec:ProblemDefinition}
%%%%%%%%%%%%%%%%%%%%%%%%%%%%%%%%%%%

This paper is motivated by applications from the US~Census Bureau,
whose goal is to release socio-demographic features of the population
grouped by census blocks, counties, and states. For instance, the
bureau is interested in releasing information such as the number of
people in a household and how many cars they own. 
This section provides a generic formalization of this release problem.

Consider a dataset $D = \{(\user_i, \unit_i, \region_i,
\unitsize_i)\}_{i=1}^n$ containing $n$ tuples $(\user_i, \unit_i,
\region_i, \unitsize_i) \in \userset \times \unitset \times \regionset
\times \unitsizeset$ denoting, respectively, a (randomly generated)
identifier for user $i \in [n]$, its \emph{unit identifier} (e.g., the
home address where she lives), the \emph{region} in which she lives,
(e.g., a census block), and a \emph{unit quantity} describing a
socio-demographic feature, e.g., the number of cars she owns, or her
salary bracket. The set of users sharing the same unit forms a
\emph{group} and $\group_\unit = \{ \user_i \in \userset \st \unit_i =
\unit\}$ denotes the group of unit $\unit$. The socio-demographic
feature of interest is the sum of the \emph{unit quantities} of a
group $\group_\unit$, i.e., $\size_{\unit} = \sum_{\user_i \in
\group_{\unit}} \unitsize_i$. 

\begin{table}[!tb]
\centering
\parbox{.475\textwidth}{
	\centering
	\resizebox{0.35\textwidth}{!}{%%
	\begin{tabular}{@{}r r r r@{}}
		\toprule
		user & unit & region & quantity\\[-1pt]
		%$\cP$ & $\cU$ & $\cR$ & $\cZ$\\[-1pt]
	    %\multicolumn{1}{c}{}
		\midrule
		$01$ & A & GA & 1 \\[-1pt]
		$02$ & B & GA & 1 \\[-1pt]
		$03$ & A & GA & 1 \\[-1pt]
		$04$ & A & GA & 1 \\[-1pt]
		$05$ & C & GA & 1 \\[-1pt]
		$06$ & D & NY & 1 \\[-1pt]
		$07$ & E & NY & 1 \\[-1pt]
		$08$ & D & NY & 1 \\[-1pt]
		$09$ & D & NY & 1 \\[-1pt]
		${10}$ & F & NY & 1 \\[-1pt]
		${11}$ & F & NY & 1 \\[-1pt]
	 %    \cmidrule{1-4}
		% $\cP$ & $\cU$ & $\cR$ & $\cZ$\\[-1pt]
	    % \multirow{6}{*}{\textbf{0.5}} 
	    \bottomrule
	\end{tabular}
	}
	\caption{\label{tab:1} Example dataset describing the user identifier $p_i \in \cP$, its unit $u_i \in \cU$, the user's region $r_i \in \cR$ and the unit quantity $z_i \in \cZ$ $(i \in [11])$.}
}
\hfill
\parbox{.475\textwidth}{
	\centering
	\resizebox{0.35\textwidth}{!}{%%
	\begin{tabular}{@{}r r l r@{}}
		\toprule
		region& $u$ & group $G_u$ & $\sigma_u$\\[-1pt]
		% $\cR$& $u$ & $G_u$ & $\sigma_u$\\[-1pt]
		\midrule
		\multirow{3}{*}{GA}
		& A & $\{01, 03, 04\}$ & $3$ \\[-1pt]
		& B & $\{02\}$ 		   & $1$ \\[-1pt]
		& C & $\{05\}$ 		   & $1$ \\[-1pt]
		\cmidrule{1-4}
		\multirow{3}{*}{NY}
		& D & $\{06, 08, 09\}$ & $3$ \\[-1pt]
		& E & $\{07\}$ 		 & $1$ \\[-1pt]
		& F & $\{10, 11\}$ 	 & $2$ \\[-1pt]
	    \bottomrule
	\end{tabular}
	}%%
	\caption{\label{tab:2} Groups $(G_u)$ and sum of unit quantities ($\sigma_u$) for the users $u \in \cU$ of Table~\ref{tab:1}.}
}
\end{table}

\begin{example}
Consider the example dataset of \Cref{tab:1}. It shows $11$ users ($p_i \in \cP$, their addresses (unit identifiers $u_i \in \cU$) , the US states in which they live (regions $r_i \in \cU$), and the associated 0/1 quantity denoting a feature of interest ($z_i \in \cZ)$. In the running example, the feature of interest is always 1, since the application is interested with the composition of the household, i.e., how many people live at the same address. Hence the \emph{groups} identify households and the sums of unit quantities represent household sizes.
The set of users, units, and regions, are, respectively: 
$\userset = \{\textsl{01}, \ldots, \textsl{11}\}$, 
$\unitset = \{A, B, C, D, E, F\}$, and
$\regionset = \{\GA, \MI, \US\}$.
The table shows that some users live in the same address (e.g., user 01, 03, and 04 all live in address $A$), identifying the users of group $G_A$, as illustrated in \Cref{tab:2}. In addition to the groups $G_u$, for all unit $u \in \cU$, \Cref{tab:2} also reports the sum of unit quantities $\sigma_u$ associated with each group $G_u$. In the example, these quantities denote the household sizes. For instance, $\sigma_A=3$ denotes that 3 users live in unit $A$ (e.g., in group $G_A$). % The example also uses $\sizeset = \{1,\ldots,5\}$.
\end{example}

In addition to the dataset, the census bureau works with a
\emph{region hierarchy} that is formalized by a tree $\cT$ of $L$
levels.  Each level $\ell \in [L]$ is associated with a set of regions
$\regionset_\ell \subseteq \regionset$, forming a partition on $D$.
Region $r'$ is a subregion of region $r$, which is denoted by $r'
\prec r$, if $r'$ is contained in $r$ and $\text{lev}(r') =
\text{lev}(r) + 1$, where $\text{lev}(r)$ denotes the level of $r$.
The root level contains a single region $r^\top$. The children of
$r$, i.e., $ch(r) = \{r' \in \regionset | r' \prec r\}$ is the set of
regions that partitions $r$ in the next level of the hierarchy and
$pa(r)$ denotes the parent of region $r$ ($r \neq r^\top$).
\Cref{fig:1} provides an illustration of a hierarchy of 2
levels. Each node represents a region. The regions \textsl{GA} and
\textsl{NY} form a partition of region \textsl{US}. 

The set $\sizeset$ of all unit sizes, 
($\{\size_\unit \st \unit \in \unitset\} \subseteq \sizeset$), also plays an important role. Indeed, the bureau is
interested in releasing, for every unit size $\size \in \sizeset$, the quantity
$\gsize_\size = | \{ \group_{\unit} \st \unit \in \unitset,
\size_{\unit} = \size \}|$, i.e., the number of groups of size
$\size$. 
The number of groups with size $\size \in \sizeset$ and region
$\region \in \regionset$ is denoted by $\gsize_\size^\region = |\{
\unit \in \unitset \st \size_\unit = \size \land \unit \in \region
\}|$ and $\gsizevec^r = (\gsize_1^r, \ldots, \gsize_{N}^r)$ denotes
the vector of \emph{group sizes} for region $r$, where $N =
|\sizeset|$. 

\begin{example}
In the running example, \Cref{fig:1} illustrates a region hierarchy, depicting \US{} as the root region, at level 1 of the hierarchy. Its children $\textsl{ch}(\US) = \{\GA, \MI\}$ represent the states of Georgia and New York, at level 2 of the hierarchy. 
\Cref{tab:3} illustrates the group size table that, for each group size $s \in [N=5]$, counts the number of groups that are households of size $s$. For instance, in region \GA, there two groups $u$ whose size $\sigma_u = 1$: They are groups A and B; in region $\MI$, is a single group whose size is equal $1$: E; Finally, in region $\US$ there three groups whose size $\sigma_u = 1$: A, B, and E (represented in the first row of the table).
The group vectors, for each region, are, respectively: 
$\gsizevec^{\GA} = (2, 0, 1, 0, 0)$, 
$\gsizevec^{\MI} = (1, 1, 1, 0, 0)$, and 
$\gsizevec^{\US} = (3, 1, 2, 0, 0)$. 
\end{example}

\begin{table}[!t]
\centering
\parbox{.45\textwidth}{
\centering
	\includegraphics[width=0.35\textwidth]{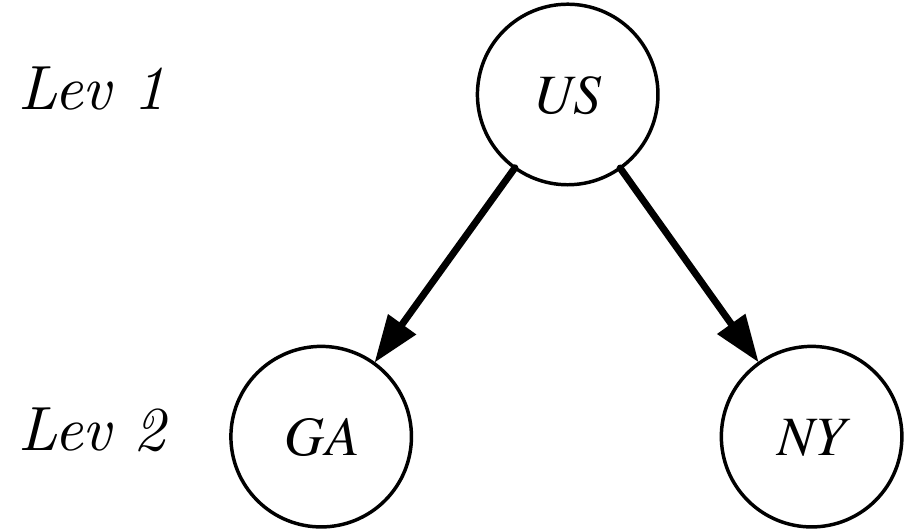}
	\captionof{figure}{\label{fig:1} 
	A region hierarchy associated to the example of \Cref{tab:1}.}
}
\hfill
\parbox{.45\textwidth}{
	\centering
	\resizebox{0.35\textwidth}{!}
	{%%
	\begin{tabular}{@{}l r r r@{}}\toprule
	\multirow{2}{*}{group sizes} &
	\multicolumn{2}{c}{Lev 2} &
	\multicolumn{1}{c}{Lev 1}\\[-1pt]
	\cmidrule{2-3}\cmidrule{4-4}
	& GA & NY & US\\[-1pt] 	 \midrule
	$1$ & 2 & 1 & 3 \\[-1pt]
	$2$ & 0 & 1 & 1 \\[-1pt]
	$3$ & 1 & 1 & 2 \\[-1pt]
	$4$ & 0 & 0 & 0 \\[-1pt]
	$5=|\cS|=N$ & 0 & 0 & 0 \\[-1pt] \cmidrule{2-4}
    & $\bm{n}^{GA}$ &$\bm{n}^{NY}$ &$\bm{n}^{US}$\\ \bottomrule
	\end{tabular}
	}
  \captionof{table}{\label{tab:3} The hierarchical group-size quantities associated to the example dataset of \Cref{tab:1}.}
}
\end{table}

It is now possible to define the problem of interest to the bureau:
{\em The goal is to release, for every group size $s \in [N]$ and
region $r \in \regionset$, the numbers $\gsize_s^r$ of groups of size
$s$ in region $r$, while preserving individual privacy.}  The region
hierarchy and the group sizes $\sizeset$ are considered public
\emph{non-sensitive} information. The entries associating users with
groups (see \Cref{tab:1}) are \emph{sensitive}
information. Therefore, the paper focuses on protecting the privacy of
such information. For simplicity, this paper assumes that the region
hierarchy has exactly $L$ levels. The paper also focuses on the vastly common case when $z_i
\in \{0,1\} \ (i \in [n])$, but the results generalize to arbitrary
$z_i$ values.

%%%%%%%%%%%%%%%%%%%%%%%%%%%%%%%%%%%
\section{Differential Privacy}
\label{sec:differential_privacy}
%%%%%%%%%%%%%%%%%%%%%%%%%%%%%%%%%%%

This paper adopts the framework of differential privacy
\cite{Dwork:06,Dwork:13}, which is the de-facto standard for privacy
protection.

\begin{definition}[Differential Privacy \cite{Dwork:06}]
	A randomized algorithm $\cM:\sD \to \sR$ is $\epsilon$-\emph{differentially private} if
	\begin{equation}
	\label{eq:dp}
		\Pr[\cM(D_1) \in O] \leq \exp(\epsilon) \Pr[\cM(D_2) \in O],  
	\end{equation}
	for any output response $O \subseteq \sR$ and any two datasets $D_1, D_2 \in \sD$ differing in at most one individual (called \emph{neighbors} and written $D_1 \sim D_2$).
\end{definition}

\noindent 
Parameter $\epsilon > 0$ is the \emph{privacy loss} of the
algorithm, with values close to $0$ denoting strong
privacy. Intuitively, the definition states that the probability of
any event does not change much when a single individual data is added
or removed to the dataset, limiting the amount of information that
the output reveals about any individual.  

This paper relies on the \emph{global sensitivity method}
\cite{Dwork:06}.  The global sensitivity $\Delta_q$ of a function $q:
\sD \to \RR^k$ (also called \emph{query}) is defined as the maximum
amount by which $q$ changes when a single individual is added to, or
removed from, a dataset: \begin{equation} \label{eq:sensitivity}
\Delta_q = \max_{D_1 \sim D_2} \| q(D_1) - q(D_2) \|_1.
\end{equation} Queries in this paper concern the group size vectors
$\gsizevec^r$ and neighboring datasets differ by the presence or
absence of at most one record (see \Cref{tab:1,tab:3}).

The global sensitivity is used to calibrate the amount of noise to add
to the query output to achieve differential privacy.  There are
several sensitivity-based mechanisms \cite{Dwork:06,Mcsherry:07} and
this paper uses the \emph{Geometric} mechanism \cite{Ghosh:12} for
\emph{integral} queries. It relies on a double-geometric distribution
and has slightly less variance than the ubiquitous Laplace mechanism
\cite{Dwork:06}.

\begin{definition}[Geometric Mechanism
\cite{Ghosh:12}] Given a dataset $D$, a query $q: \sD \to \RR^k$, and
$\epsilon>0$, the geometric mechanism adds independent noise to each
dimension of the query output $q(D)$ using the distribution
\[
P(X\!=\! v) \!=\! \frac{1 - e^{-\epsilon}}{1 + e^{-\epsilon}} 
\,e^{\left(-\epsilon \frac{|v|}{\Delta_q}\right)}.
\]
\end{definition}

\noindent This distribution is also referred to as
\emph{double-geometric} with scale $\Delta_q / \epsilon$.  In the
following, $Geom(\lambda)^k$ denotes the i.i.d.~double-geometric
distribution over $k$ dimensions with parameter $\lambda$. The
geometric mechanism satisfies $\epsilon$-differential privacy
\cite{Ghosh:12}. Differential privacy also satisfies several important properties \cite{Dwork:13}.

\begin{lemma}[Sequential Composition]
\label{th:seq_composition}
The composition of two $\epsilon$-differentially private mechanisms ($\cM_1$, $\cM_2$) satisfies $2\epsilon$-differential privacy.
\end{lemma}

\begin{lemma}[Parallel Composition] 
\label{th:par_composition} 
Let $D_1$ and $D_2$ be disjoint subsets of $D$ and $\cM$ be an
$\epsilon$-differential private algorithm.  Computing $\cM(D
\cap D_1)$ and $\cM(D \cap D_2)$ satisfies $\epsilon$-differential privacy.
\end{lemma}

\begin{lemma}[Post-Processing Immunity] \label{th:postprocessing}
Let $\cM$ be an $\epsilon$-differential private algorithm and $g$ be
an arbitrary mapping from the set of possible output sequences $O$
to an arbitrary set. Then, $g \circ \cM$ is $\epsilon$-differential
private.  
\end{lemma}

\section{The Privacy-Preserving Group Size Release Problem}
\label{sec:group_est}

This section formalizes the Privacy-preserving Group Size Release
(PGSR) problem \cite{kuo2018_VLDB}.  Consider a
dataset $D$, a region hierarchy $\cT$ for $D$, where each node
$\node^r$ in $\cT$ is associated with a vector $\gsizevec^r \in
\mathbb{Z}_{+}^N$ describing the group sizes for region $r \in
\regionset$, and let $G = \sum_{s \in [N]} \gsize^\top_s$ be the total
number of individual groups in $D$, which is public information (see
\Cref{fig:hierarhcy} for an example).  The PGSR problem consists in
releasing a hierarchy of group sizes $\tilde{\cT} = \langle
\tilde{\gsizevec}^r \st r \in \regionset \rangle$\footnote{We abuse
notation and use the angular parenthesis to denote a hierarchy.} that
satisfies the following conditions:

\begin{enumerate}[topsep=1ex,itemsep=-1ex,partopsep=1ex,parsep=1ex]%,label=\textit{\roman*}]
	\item \label{gse:1}
	\textsl{Privacy}: $\tilde{\cT}$ is $\epsilon$-differentially private.
	\item \label{gse:2}
	\textsl{Consistency}: 
	For each region $r \in \regionset$ and group size $s \in \sizeset$,
	the group sizes in the subregions $r'$ of $r$ add up to those in
	region $r$: 
	\[
			\tilde{\gsize}^r_s = \sum_{{r'} \in ch(r)} 
			\tilde{\gsize}^{r'}_s.
	\]
	\item \label{gse:3}
	\textsl{Validity}: The values $\tilde{\gsize}_s^r$ are non-negative 
	integers.
	\item \label{gse:4}
	\textsl{Faithfulness}: The group sizes at each level $\ell$ of the 
	hierarchy add up to the value $G$:
	\[
		\sum_{r \in \regionset_\ell} \sum_{s \in [N]} 
		\tilde{\gsize}^r_s = G.
	\]
\end{enumerate}
These constraints ensure that the hierarchical group size estimates 
satisfy all publicly known properties of the original data.

\begin{figure}[!t]
	\centering
	\hspace{-18pt}
	\includegraphics[width=0.5\linewidth]{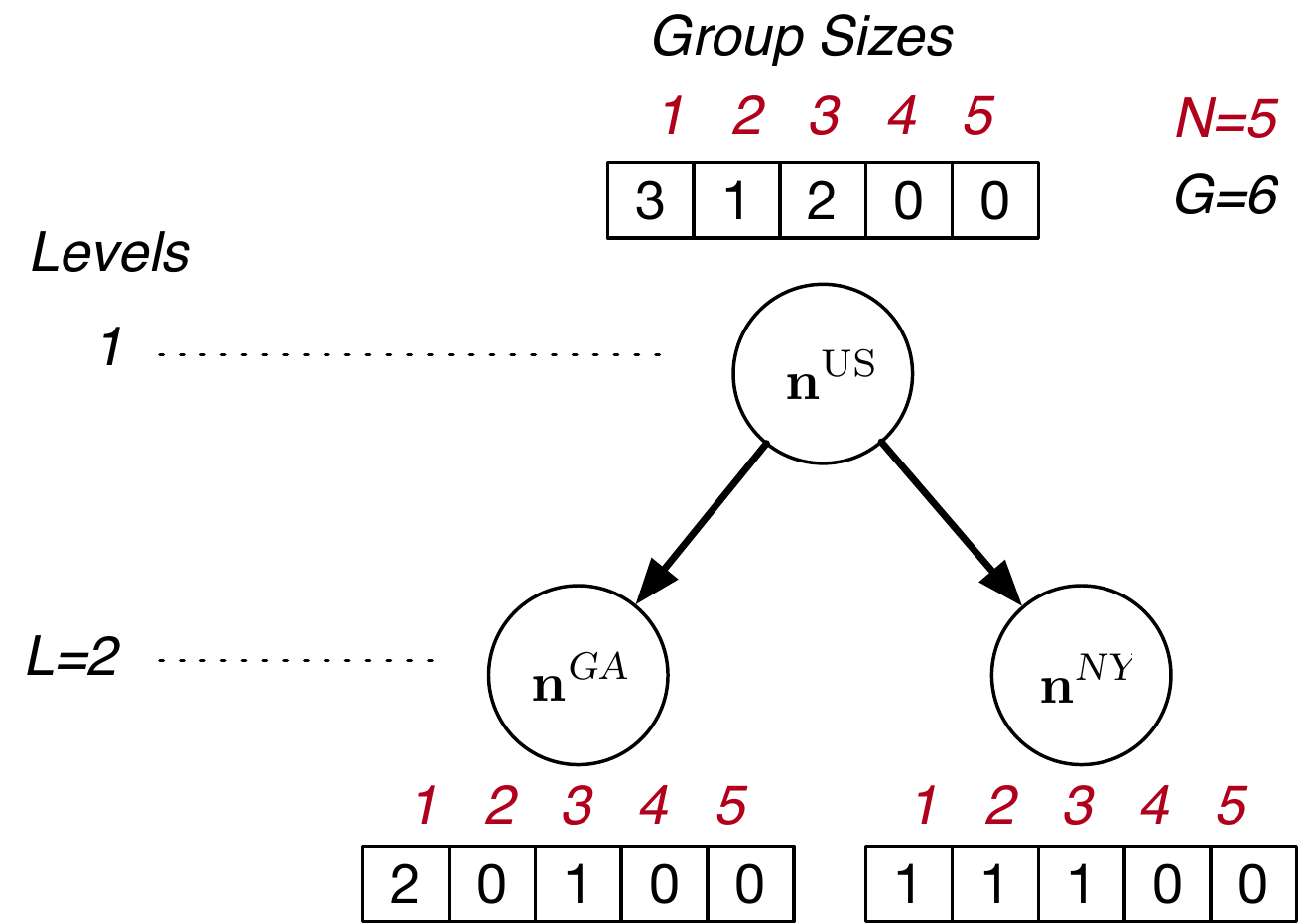}
	\caption{\label{fig:hierarhcy} 
	\emph{Group size} hierarchy $\cT_2$ associated with the dataset of \Cref{tab:1}.}
	% \emph{Cumulative group size} hierarchy $\cT_2^c$ (b).}
\end{figure}

\section{The Direct Optimization-Based PGSR Mechanism}
\label{sec:IPPGSR}

This section presents a two-step mechanism for the PGSR problem, first 
introduced in \cite{kuo2018_VLDB}. The first 
step produces a noisy version of the group sizes, whereas the second
step restores the feasibility of the PGSR constraints while staying
as close as possible to the noisy counts. The first step produces a
noisy hierarchy $\tilde{\cT} = \{\tilde{\gsizevec}^r \st r \in
\regionset \}$ using the geometric mechanism with parameter
$\lambda=\frac{2L}{\epsilon}$ on the vectors $\gsizevec^r$:

\begin{equation}
	\label{eq:geom_q}
	\tilde{\gsizevec}^r = \gsizevec^r + \textsl{Geom}\Big(
		\frac{2L}{\epsilon}\Big)^N.
\end{equation}

\noindent
The following lemma and theorem, whose proofs are in the appendix, show that 
this step satisfies $\epsilon$-differential privacy. 
\begin{lemma}
The sensitivity $\Delta_{\gsizevec}$ of the group estimate query is 2.
\end{lemma}

\begin{theorem}
	\label{th:privacy_H}
	$\cM_H$ satisfies $\epsilon$-differential privacy.
\end{theorem}

\noindent 
Proofs for all Theorems and Lemmas are reported in the appendix. 

% ------------------------------------ %
% \begin{wrapfigure}[9]{r}{195pt} 
% \vspace{-10pt}
% \begin{mdframed}
% \vspace{-10pt}
% {\small
% \begin{align}
% 	\hspace{-10pt}
% 	\displaystyle{\minimize_{ \{\hat{\gsizevec}^r\}_{r \in \regionset}}} \;\; 
% 		\sum_{r \in \regionset} \| &\hat{\gsizevec}^r - \tilde{\gsizevec}^r \|_2^2 
% 		\tag{H1} \label{eq:H1}\\
% 	%
% 	\mbox{s.t:} \;\; 
% 		%%%%%%%%%%
% 		\sum_{s \in [N]} \hat{\gsize}^{r}_s &= G 
% 		\;\; \forall r \in \regionset
% 		\tag{H2} \label{eq:H2} \\
% 		%%%%%%%%%%
% 		\sum_{c \in ch(r)}  \hat{\gsize}_s^{c} &= \hat{\gsize}_s^r 
% 		\;\;  \forall r \in \regionset, s \in [N] 
% 		\tag{H3} \label{eq:H3}\\
% 		%%%%%%%%%%
% 		\hat{\gsize}^r_s \in &D^r_s %\{0, 1&, \ldots\} 
% 		\;\;\;\;\;\; \forall r \in \regionset, s \in [N] 
% 		\tag{H4} \label{eq:H4}
% \end{align}
% }
% \vspace{-18pt} 
% \end{mdframed}
% \caption{The $\cM_H$ post-processing step. \label{fig:NP}}
% \end{wrapfigure} %

The output of the first step satisfies Condition \ref{gse:1} of the
PGSR problem but it will violate (with high probability) the other
conditions. To restore feasibility, this paper uses a post-processing
strategy similar to the one proposed in \cite{fioretto:AAMAS-18} for
mobility applications and in \cite{kuo2018_VLDB}.  After generating
$\tilde{\cT}$ using \Cref{eq:geom_q}, the mechanism post-processes the
values $\tilde{\gsizevec}^r$ of $\tilde{\cT}$ through the
\emph{Quadratic Integer Program (QIP)} depicted in Model \ref{fig:NP}.
 Its goal is to find a new region hierarchy $\hat{\cT}$, optimizing
over the variables $\hat{\gsizevec}^r \!=\! (\hat{\gsize}_1^r\ldots
\hat{\gsize}_N^r)$ for each $r \!\in\! \regionset$, so that their
values stay close to the noisy counts of the first step, while
satisfying faithfulness (Constraint \eqref{eq:H2}), consistency
(Constraint \eqref{eq:H3}), and validity (Constraint \eqref{eq:H4}).
In the optimization model, $D_s^r$ represents the domain (of integer,
non-negative values) of $\hat{\gsize}^r_s$. The resulting mechanism is
called the \emph{Hierarchical PGSR} and denoted by $\cM_H$. It
satisfies $\epsilon$-differential privacy because of post-processing
immunity (\Cref{th:postprocessing}), since the
post-processing step of $\cM_H$ uses exclusively differentially
private information ($\tilde{\cT}$).
% ------------------------------------  %
\begin{model}[!tb]
\caption{The $\cM_H$ post-processing step \label{fig:NP}}
\vspace{-8pt}
\begin{alignat}{2}
	\minimize_{ \left\{ \hat{\gsizevec}^r \right\}_{r \in \regionset}}\;\; 
		& 
		\sum_{r \in \regionset} \left| \hat{\gsizevec}^r - \tilde{\gsizevec}^r \right|_2^2 
		&&
		\tag{H1} \label{eq:H1}\\
	\subjectto\;\;&
		%%%%%%%%%%
		\sum_{r \in \regionset_\ell} \sum_{s \in [N]} 
		\hat{\gsize}^r_s = G
		&& \forall \ell \in [L]
		\tag{H2} \label{eq:H2} \\
		%%%%%%%%%%
		&
		\sum_{c \in ch(r)}  \hat{\gsize}_s^{c} = \hat{\gsize}_s^r 
		\qquad
		&&  \forall r \in \regionset, s \in [N] 
		\tag{H3} \label{eq:H3}\\
		%%%%%%%%%%
		&
		\hat{\gsize}^r_s \in D^r_s %\{0, 1&, \ldots\} 
		&& \forall r \in \regionset, s \in [N] 
		\tag{H4} \label{eq:H4}
\end{alignat}
\vspace{-16pt}
\end{model}

Solving this QIP is intractable for the datasets of interest to the
census bureau. Therefore, the experimental results consider a version
of $\cM_H$ that \emph{relaxes the integrability constraint}
\eqref{eq:H4} and rounds the solutions. However, the resulting optimization
problem becomes convex but presents two limitations: \emph{(i)} its
final solution $\hat{T}$ may violate the PGSR \emph{consistency}
(\ref{gse:2}) and \emph{faithfulness} (\ref{gse:4}) conditions, and
\emph{(ii)} the mechanism is still too slow for very large problems.

\section{The Dynamic Programming PGSR Mechanism}
\label{sec:dp_gse}

To overcome the $\cM_H$ limitations discussed above, this section
proposes a dynamic-programming approach for the post-processing step
and its convex relaxation. The resulting mechanism is called the
\emph{Dynamic Programming PGSR} mechanism and denoted by
$\cM_H^{\textsl{dp}}$. The dynamic program relies on a new hierarchy
$\cT^{\text{dp}}$ that modifies the original region hierarchy $\cT$ as
follows.  It creates as many subtrees as the number $N$ of group sizes 
$\sizeset$. 
In each of these subtrees, node $a^r_s$ is associated with the number 
$\gsize_s^r$ of groups of size $s$ in region $r$. Its children $\{ 
\node_s^{c} \}_{c \in ch(r)}$ are associated with the numbers 
$\gsize_s^{c}$, and so on.  
Thus, the nodes of subtree $s$ represent the groups of size $s$
for all the regions in $\regionset$. 
Finally, the new hierarchy has a root note $\node^\top$ that represents 
the total number $G$ of groups: It is associated with a \emph{dummy} 
region $\top$ whose children are the root nodes of the $N$ subtrees 
introduced above. The resulting region hierarchy is denoted $\treedp$. 

\begin{example} 
The region hierarchy $\treedp$ associated with the
running example is shown in \Cref{img:tree_regions}.  The root node
$\node^\top$ is associated with the total number of groups in $D$,
i.e., $G=6$. Its children $\node_1^{\US},\dots,\node_5^{\US}$
represent the group sizes for the root of the region hierarchy for
each group size $s \in [N=5]$.  Subtree 1, rooted at $\node_1^{\US}$,
has two children: $\node_1^{\GA}$ and $\node_1^{\MI}$, representing
the number of groups of size $1$: $\gsize_1^{\GA}$ and
$\gsize_1^{\MI}$.  The figure illustrates the association of each 
node $\node_s^r$ with its \emph{real} group size $\gsize_s^r$ (in red)
and its noisy group size generated by the geometrical mechanism (in
blue and parenthesis). 
\end{example}

Note that \emph{(i)} the value of a node equals to the sum
of the values of its children, \emph{(ii)} the group sizes
at a given level add up to $G$, and \emph{(iii)} the PGSR consistency
conditions of the nodes in a subtree are independent of those of
other subtrees.  These observations allow us to develop a dynamic
program that guarantees the PGSR conditions and exploits the
independence of each subtree associated with groups of size $s$ to
solve the post-processing problem efficiently.

For notational simplicity, the presentation omits the subscripts
denoting the group size $s$ and focuses on the computation of a single
subtree representing a group of size $s$.  The dynamic program
associates a \emph{cost table} $\ctable^r$ with each node $\node^r$ of
$\treedp$.  The cost table represents a function $\ctable^r: \dom^r
\to \mathbb{R}_+$ that maps values (i.e., group sizes) to costs, where
$\dom^r$ is the \emph{domain} (a set of natural numbers) of region
$r$. Intuitively, $\ctable^r(v)$ is the optimal cost for the
post-processed group sizes in the subtree rooted at $\node^r$ when
its post-processed group size is equal to $v$, i.e., $\hat{\gsize}^r =
v$. {\em The key insight of the dynamic program is the observation
  that the optimal cost for $\ctable^r(v)$ can be computed from the
  cost tables $\ctable^c$ of each of its children $c \in
  \textsl{ch}(r)$ using the following}:% given in \Cref{fig:bottom_up}.} 
% ------------------------------------  %
\begin{subequations}
\label{eq:d}
\begin{alignat}{3}
  	\ctable^r(v) = 
		&\Big(v - \tilde{\gsize}^r \Big)^2 
		&& + 
		&& %\tag{d1} 
		\label{eq:d1}\\
		%%%%%%%%%%%%%
		& \hspace{20pt} \ctablech^r(v) 				   
		&&= \minimize_{ \{x_c \}_{c \in \textsl{ch}(r)} } 
		&& \sum_{c \in \textsl{ch}(r)} \ctable^{c}(x_c) 
		%\tag{d2} 
		\label{eq:d2} \\
		%%%%%%%%%%%%%
    	& 
    	&& \subjectto 
    	&& \sum_{c \in \textsl{ch}(r)} x_c = v 
	   	%\tag{d3} 
	   	\label{eq:d3} \\
	   	%%%%%%%%%%%%%
		& 
		&&
		&& x_c \in D^c\qquad \forall {c \in \textsl{ch}(r)}.
		%\tag{d4} 
		\label{eq:d4}
\end{alignat}
\end{subequations}

\begin{figure}[!tbh]
	{\centering\includegraphics[width=\linewidth]{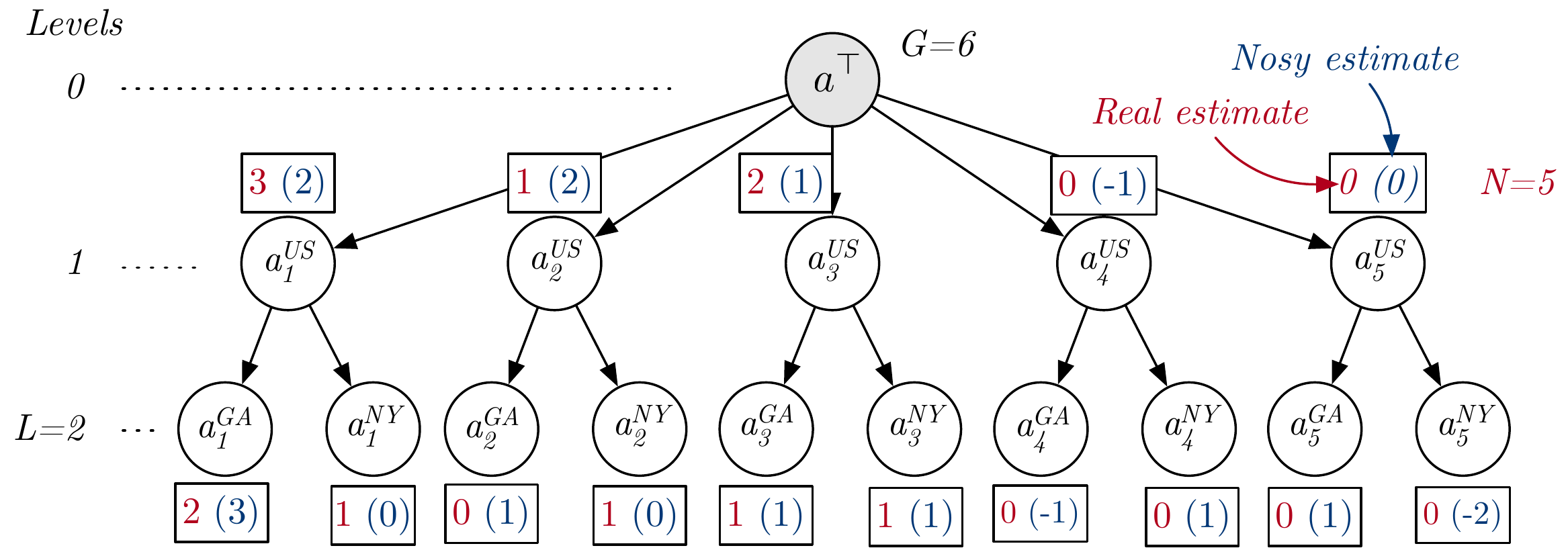}}\\
	\caption{\label{img:tree_regions}Region hierarchy $\treedp$
associated with the dataset of \Cref{tab:1}.}
\end{figure}

% subsection subsection_name (end)}
% ------------------------------------  %
In the above, \eqref{eq:d1} describes the cost for $v$ of deviating
from the noisy group size $\tilde{\gsize}^r$. The function
$\ctablech^r(v)$, defined in \eqref{eq:d2}, \eqref{eq:d3}, and
\eqref{eq:d4}, uses the cost table of the children of $r$ to find the
combination of post-processed group sizes $\{x_c \in D^c\}_{c \in
  \textsl{ch}(r)}$ of $r$'s children that is consistent \eqref{eq:d3}
and minimizes the sum of their costs \eqref{eq:d2}.

% ------------------------------------ %
% \begin{wrapfigure}[8]{r}{165pt}
% \vspace{-10pt}
% \begin{mdframed}
% \vspace{-18pt}
% {\small
% % ------------------------------------  %
% \begin{align}
% \hspace*{-12pt}
%   	\ctable^r(v) = 
% 		&\Big(v - \tilde{\gsize}^r \Big)^2 + 
% 		\tag{d1} \label{eq:d1}\\
% 		%%%%%%%%%%%%%
% 		& \hspace{-17pt} \ctablech^r(v) \!= \!\!\!\!\!
% 		\min_{ \{x_c \}_{c \in \textsl{ch}(r)} } 
% 		\!\!
% 		\sum_{c \in \textsl{ch}(r)} \!\!\! \ctable^{c}(x_c) 
% 		\tag{d2} \label{eq:d2} \\
% 		%%%%%%%%%%%%%
%     	& \quad \text{s.t.~} \!\!\!\!
% 	   		\sum_{c \in \textsl{ch}(r)} x_c = v 
% 	   		\tag{d3} \label{eq:d3} \\
% 	   	%%%%%%%%%%%%%
% 		& \phantom{\,\, s.t.~} 
% 			x_c \in D^c\;\; \forall {c \in \textsl{ch}(r)} 
% 			\tag{d4} \label{eq:d4}
% \end{align}
% }
% \vspace{-26pt} 
% \end{mdframed}
% \caption{The $\cM_H^{\textsl{dp}}$ Bottom-up step. \label{fig:bottom_up}}
% \end{wrapfigure} %

\begin{wrapfigure}[10]{R}{0.5\textwidth}%
\vspace{-8pt}
\centering
\resizebox{0.99\linewidth}{!}
{%
	\begin{tabular}{@{}r |r r r@{}l@{}}
	\toprule
	$\bm{v}$ & $\tau_1^{\GA}$ & $\tau_1^\MI$ & 
	\multicolumn{2}{l}{~~$\tau_1^\US$} \\[-1pt]
	\midrule
	$\bm{0}$ & $9$   	   & ${}^\star0$  & $4$         &  $\!+\!\min(9\!+\!0)$\\[-1pt]
	$\bm{1}$ & $4$   	   & $1$          & $1$         &  $\!+\!\min(9\!+\!1, 4\!+\!0)$\\[-1pt]
	$\bm{2}$ & $1$	 	   & $4$          & $0$         &  $\!+\!\min(0\!+\!4, 4\!+\!1, 1\!+\!0)$\\[-1pt]
	$\bm{3}$ & ${}^\star0$ & $9$          & ${}^\star1$ &  $\!+\!\min(0\!+\!0, 1\!+\!1, 4\!+\!4, 9\!+\!9)$\\[-1pt]
	$\bm{4}$ & $1$         & $16$         & $4$         &  $\!+\!\min(1\!+\!0, 0\!+\!1, 1\!+\!4, \ldots)$\\[-1pt]
	\bottomrule
\end{tabular}
}%
\captionof{table}{\label{tab:Tdp}Example of cost table computation for subtree 
associated to the group 1 estimates.}
\end{wrapfigure} %
The dynamic program exploits these concepts in two phases. The first
phase is bottom-up and computes the cost tables for each node,
starting from the leaves only, which are defined by \eqref{eq:d1}, 
and moving up, level by level, to the root. The cost table at the
root is then used to retrieve the optimal cost of the problem. The
second phase is top-down: Starting from the root, each node $a^r$
receives its post-processed group size $\hat{n}^r$ and solves
$\ctablech^r(\hat{n}^r)$ to retrieve the optimal post-processed group
sizes $\hat{n}^c = x_c$ for each child $c \in \textsl{ch}(r)$.

An illustration of the process for the running example is
illustrated in \Cref{tab:Tdp}. It depicts the cost tables
$\ctable_1^{\GA}, \ctable_1^{\MI}$, and $\ctable_1^{\US}$ related to
the subtree rooted at $\node_1^{\US}$ (groups of size 1) computed
during the bottom-up phase. The values selected during the top-down
phase are highlighted with a star symbol.

In the implementation, the values $\ctablech^r(v)$ are computed using
a constraint program where \eqref{eq:d1} uses a table
constraint. The number of optimization problems in the dynamic program
is given by the following theorem.

\begin{theorem} 
\label{th:treedp_compl}
Constructing $\hat{\cT}^{\textsl{dp}}$ requires
solving $O(|\regionset| N \bar{D}$) optimization problems given
in \Cref{eq:d}, where $\bar{D} = \max_{s,r} |D_s^r|$ for
$r\in \regionset, s\in [N]$.  
\end{theorem}

%%%%%%%%%%%%%%%%%%%%%%%%%
\section{A Polynomial-Time PGSR Mechanism}
\label{sec:dp_gse_exploit}
%%%%%%%%%%%%%%%%%%%%%%%%%

The dynamic program relies on solving an optimization problem for each
region. This section shows that this optimization problem can be
solved in polynomial time by exploiting the structure of the cost
tables.

A cost table is a finite set of pairs ($s,c$) where $s$ is a group
size and $c$ is a cost. When the pairs are ordered by increasing
values of $s$ and line segments are used to connect them as in
\Cref{fig:ex2}, the resulting function is Piecewise Linear (PWL). For
simplicity, we say that a cost table is PLW if its underlying function
is PWL. Observe also that, at a leaf, the cost table is Convex PWL
(CPWL), since the L2-Norm is convex (see \Cref{eq:d1}).

{\em The key insight behind the polynomial-time mechanism is the
recognition that the function $\ctablech^r$ is CPWL whenever the cost
tables of its children are CPWL.} As a result, by induction, the cost
table of every node $a^r$ is CPWL.

\begin{lemma}
\label{lm:convexity}
The cost table $\ctable_s^r$ of each node $a_s^r$ of $\treedp$ is CPWL.
\end{lemma}

\noindent
\Cref{lm:convexity} makes it possible to design a polynomial-time
algorithm to replace the constraint program used in the dynamic
program. The next paragraphs give the intuition underlying the
algorithm. 

\begin{figure}[!t]
	\centering\includegraphics[width=\linewidth]{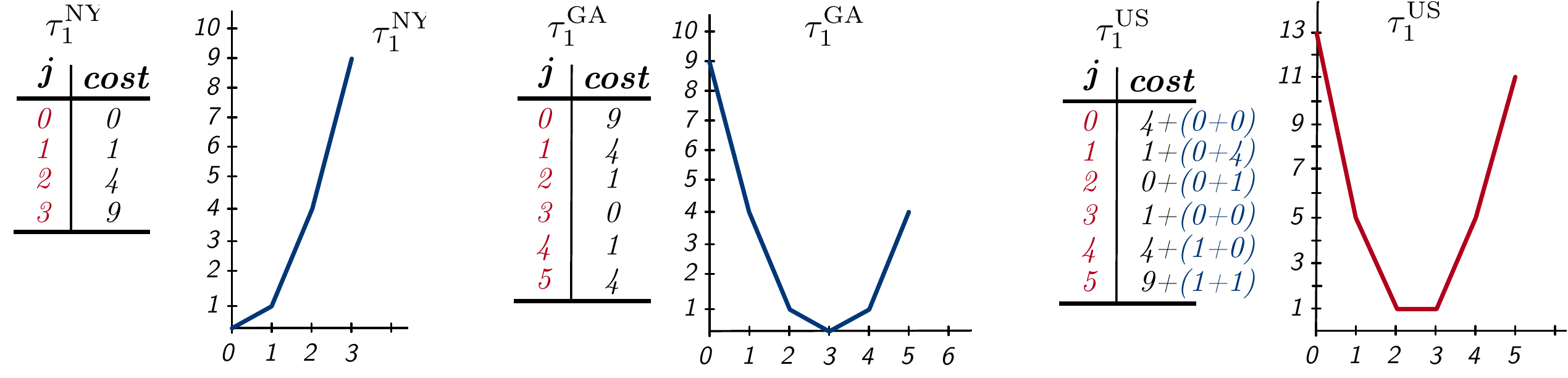}\\
	\caption{\label{fig:ex2} 
	Cost tables, extending \Cref{img:tree_regions}, computed by the mechanism.}
\end{figure}

Given a node $a^r$, the first step of the algorithm is to select, for
each node $c \in \textsl{ch}(r)$, the value $v_c^0$ with minimum cost,
i.e., $v^0_c = \argmin_v \ctable^c(v)$. As a result, the value $V^0 =
\sum_{c} v^0_c$ has minimal cost $\ctablech^r(V^0) = \sum_c
\ctable^c(v^0_c)$.  Having constructed the minimum value in cost table
$\ctablech^r$, it remains to compute the costs of all values $V^0+k$
for all integer $k \in [1, \max D^r - V^0]$ and all values $V^0-k$ for
all integer $k \in [1, V^0 - \min D^r]$. The presentation focuses on
the values $V^0+k$ since the two cases are similar. Let $\bv^0 = \{
v^0_c \}_{c \in \textsl{ch}(r)}$. The algorithm builds a sequence of
vectors $\bv^1,\bv^2,\ldots,\bv^k,\ldots$ that provides the optimal
combinations of values for $\ctablech^r(V^0+1),\ctablech^r(V^0+2),\ldots,\ctablech^r(V^0+k),\ldots$.
Vector $\bv^k$ is obtained from $\bv^{k-1}$ by changing the value of a single child whose cost table has the smallest slope, i.e.,
\begin{align}
\label{eq:special_v}
v_c^{k} = \left\{ 
\begin{array}{l l}
v_c^{k-1} + 1 & \mbox{if }\; c = \argmin_c 
		\ctable^c(v_c^{k-1} + 1) - \ctable^c(v_c^{k-1}) \\ 
v_c^{k-1}     & \mbox{otherwise.} 
\end{array}\right.
\end{align}

\noindent
Once $\ctablech^r$ has been computed, cost table $\ctable^r$ can be
computed easily since both $(v - \tilde{n}^r)^2$ and $\ctablech^r$ are CPWL and the sum of two CPWL functions is CPWL. A full algorithm is described in the next section.

\begin{example} These concepts are illustrated in \Cref{fig:ex2},
where the values of $\ctablech^r$ are highlighted in blue, and in parenthesis, in the right table. The first step identifies that $v^0_{\MI} = 0$, $v^0_{\GA}=3$, thus 
$V^0 = 3$ and $\ctablech^r(3) = \ctable^{\MI}(0) + \ctable^{\GA}(3) =
0 + 0 = 0$.  \noindent In the example, $\bv^0 = (v^0_{\MI},v^0_{\GA})
= (0,3)$ and $\bv^1 = (v^1_{\MI}, v^1_{\GA}) = (1,3)$, 
since $\MI = \argmin\{ \ctable^{\MI}(0+1) - \ctable^{\MI}(0),
					    \ctable^{\GA}(3+1) - \ctable^{\GA}(3)\}
			= \argmin\{1-0, 1-0\}$. Its associated cost, 
			$\ctablech^{\US}(\bv^1) = 1+0 =1$.
$\bv^2 = (1,4)$ since $\GA = \argmin \{ (\ctable^{\MI}(1+1) - \ctable^{\MI}(1)),
(\ctable^{\GA}(3+1) - \ctable^{\GA}(3)) \} = \argmin \{ (4-1), (1-0)
\}$, and its associated cost $\ctablech^{\US}(\bv^2) = 1 + 1 = 2$.
\end{example}

\begin{theorem}
\label{tm:convexity}
The cost table $\ctable_s^r$ of each node $a_s^r$ of $\treedp$ is CPWL.
\end{theorem}

\subsection{Computing Cost Tables Efficiently}
\label{appendix:algorithm}

	\begin{algorithm}[!t]
		\caption{TableMerge($\node^r$)}
		\label{alg:merge}

		\DontPrintSemicolon
		\SetKwFunction{FMerge}{Merge}
		\SetKwProg{Fn}{Function}{:}{}

		$\bv^+ \!\gets\! \textsl{sort}^2(\langle c, 
		\ctable^c(v) - \ctable^c(v-1) \rangle \st c \!\in\! \textsl{ch}(r), v \!\in\! \dom^{c+})
		\!\!\!\!\!\!\!\!\!\!\!\!\!\!\!$\;
		$\bv^- \!\gets\! \textsl{sort}^2(\langle c, 
		\ctable^c(v) - \ctable^c(v-1)
		\rangle \st c \!\in\! \textsl{ch}(r), v \!\in\! \dom^{c-})
		\!\!\!\!\!\!\!\!\!\!\!\!\!\!\!$\;
		$\ctable^r \gets \FMerge(\textsl{extract}^1(\bv^+), +1, \ctable^r)$\; 
		$\ctable^r \gets \FMerge(\textsl{extract}^1(\bv^-), -1, \ctable^r)$
		\vspace{4pt}

		\Fn{\FMerge{$\bv$, $\kappa$, $\ctable^r$}}{
	  		$\bv[c] \gets \argmin \ctable^c \;\; \forall c \in \textsl{ch}(c)$\;
	  		$v \gets  \sum_{c \in \textsl{ch}(r)} \bv[c]$\;
	  		\lIf{$\kappa > 0 \land v \in D^r$}{
	        	$\ctable^r(v) \gets \sum_{c \in \textsl{ch}(c)} \ctable^c(\bv[c])$
	        }
	        \For{$e \in \bv$}{
	        	$\bv[e] \gets \bv[e] + \kappa$\;
	        	$v \gets  \sum_{c \in \textsl{ch}(r)} \bv[c]$\;
	        	\lIf{$v \in D^r$} {
	        	$\ctable^r(v) \gets 
	        		\sum_{c \in \textsl{ch}(r)} \ctable^c(\bv[c])$
	        	}
	        }
	        \KwRet $\ctable^r$
		}
	\end{algorithm}
%\end{wrapfigure}

The procedure to compute $\ctablech_s^r$ is depicted in
\Cref{alg:merge}: It is executed during the bottom-up phase of the dynamic
program in lieu of \Cref{eq:d1,eq:d3,eq:d2,eq:d4}.  It takes as input
the node $a^r$ and returns its associated cost function $\ctable^r$.  Let 
$D^{c+} = \{v \st v \in D^c \land v > \argmin_{v'} \ctable^c(v')\}$ be the set of elements in the domain of $\ctable^c$ whose values are greater than the value corresponding to the table minimum cost.
Similarly, let $D^{c-} = \{v \st v \in D^c \land v < \argmin_{v'} \ctable^c(v')\}$.  Lines $1$ and $2$ construct the vectors
$\bv^+$ and $\bv^-$ that list the pairs $(c, \ctable^c(v) -
\ctable^c(v-1))$ of node identifier and \emph{slope} associated with
the cost function for every child node $\node^c$ of $\node^r$ and
element $v \in D^{c+}$ (resp.~$\in D^{c-}$), sorted using the second
element of the pair.  Lines $3$ and $4$ extract the identifiers
for each element of the sorted lists $\bv^+$ and $\bv^-$ and call the
function $\FMerge$ to update the node cost function $\ctable^r$.

The heart of \Cref{alg:merge} is the function $\FMerge$. For a node
$\ctable^r$, it takes as input a sorted vector of the identifiers
(containing as many elements as the sum of all children's domain
sizes), a value $\kappa \!\in\! \{+1, -1\}$, and the current node
$\ctable^r$.  Line $6$ constructs a map $\bv$ that assigns to each
children $\node^c$ of $\node^r$ the value associated with the minimum
cost of its table $\ctable^r$.  Line $7$ sums the values in the map
$\bv$ resulting in a value $v$ for which the function will compute the
cost $\ctable^r$ of the parent node (line $8$). The conditional
statement ensures that this operation is done only once--the $\FMerge
$ function is also called on vector $\bv^-$.  Next, for each element
in the sorted vector $\bv$, the function selects the next element $e$,
it increases its current index $\bv[e]$ by $\kappa$, and repeats the
operations above, effectively computing the value $\ctable^r(v)$ for
the cost table of node $\node^r$. Line $12$ simply ensures that the
computed value is in the domain of the node.

\begin{example} 
An example of the effect of \Cref{alg:merge} executed on a few steps 
of the cost tables $\ctable_1^{\US}, \ctable_1^{\GA},
\ctable_1^{\MI}$ is illustrated in \Cref{fig:ex2}.  Consider the
computation of the cost function $\ctable_1^{\textsl{\US}}$. We focus
only on the first call of the $\FMerge$ function.  The algorithm first
constructs the vector $\bv^+ = ((\MI,1), (\GA,4), (\MI,2), (\GA,5),
(\MI,3))$ (line 1), extracts its first component $\bv^+ =
(\MI, \GA, \MI, \GA, \MI)$ (line 3), and hence calls the $\FMerge$
function.  The routine first initializes the vector $\bv$ as $\bv[\MI]
= \argmin \ctable_1^{\MI} = 0, \bv[\GA] = \argmin \ctable_1^{\GA} = 3$
(line 6), and $v = \bv[\MI] + \bv[\GA] = 0 + 3 = 3$ (line 7).  It
hence computes the value $\ctable_1^{\US}(3) = 1$ (line 8).  Its
associated cost is composed by the parent contribution (from \Cref{eq:d1})
$|3 - 2|^2 = 1$ and the children contribution $\ctable_1^{\MI}(0) = 0$
and $\ctable_1^{\GA}(3)=0$. 
Next, the algorithm increases $v$ by
$1$ ($v=4$), in line 11, and selects the next element in $\bv^+$
(i.e., "\MI"), in line 9.  It increments its value, $\bv[\MI] = 1$
(line 10), and computes: $\ctable_1^{\US}(4) = |4 - 2|^2
+ \ctable_1^{\MI}(\bv[\MI]) + \ctable_1^{\GA}(\bv[\GA]) = 4 + 1 + 0 =
5$ (line 12).  Finally, it increases $v$ by $1$ ($v=5$) and selects
the next element in $\bv^+$ (i.e., ``\GA"). It increments its value,
$\bv[\GA] = 1$ and computes: $\ctable_1^{\US}(5) = |5 - 2|^2
+ \ctable_1^{\MI}(\bv[\MI]) +
\ctable_1^{\GA}(\bv[\GA]) = 9 + 1 + 1 = 11$.  
\end{example}

\begin{theorem}\label{tm:dp_efficiency}
The cost table $\ctable^r_s$ for each region $r$ and size $s$ can be computed in time $O\big(\bar{D} \log \bar{D} \big)$.
%\Cref{alg:merge} requires requires $O\big(\bar{D} \log \bar{D} \big)$ operations, where $\bar{D} = \max_{s,r} \bar{D}_s^r$ for $r\in \regionset, s\in [N]$.
\end{theorem}

The result above can be derived observing that the runtime complexity
of \Cref{alg:merge} is dominated by the sorting operations in lines
$1$ and $2$.

\section{Cumulative Counts for Reduced Sensitivity}
\label{sec:cumulative_mechanisms}

In the mechanisms presented so far, each query has sensitivity
$\Delta_{\gsizevec} = 2$. This section exploits the structure of the
group query to reduce the query sensitivity and thus the noise
introduced by the geometric mechanism. The idea relies on an operator
$\oplus : \mathbb{Z}_+^N \to \mathbb{Z}_+^N$ that, given a vector
$\gsizevec = (\gsize_1, \ldots, \gsize_N)$ of group sizes, returns its
cumulative version $\csizevec = (\csize_1, \ldots, \csize_N)$ where
$\csize_s = \sum_{k=1}^s \gsize_k$ is the cumulative sum of the first
$s$ elements of $\gsizevec$. 

\begin{lemma}
The sensitivity $\Delta_{\csizevec}$ of the cumulative group estimate query is 1.
\end{lemma}

\begin{figure}[!tb]
	\centering\includegraphics[width=0.5\linewidth]{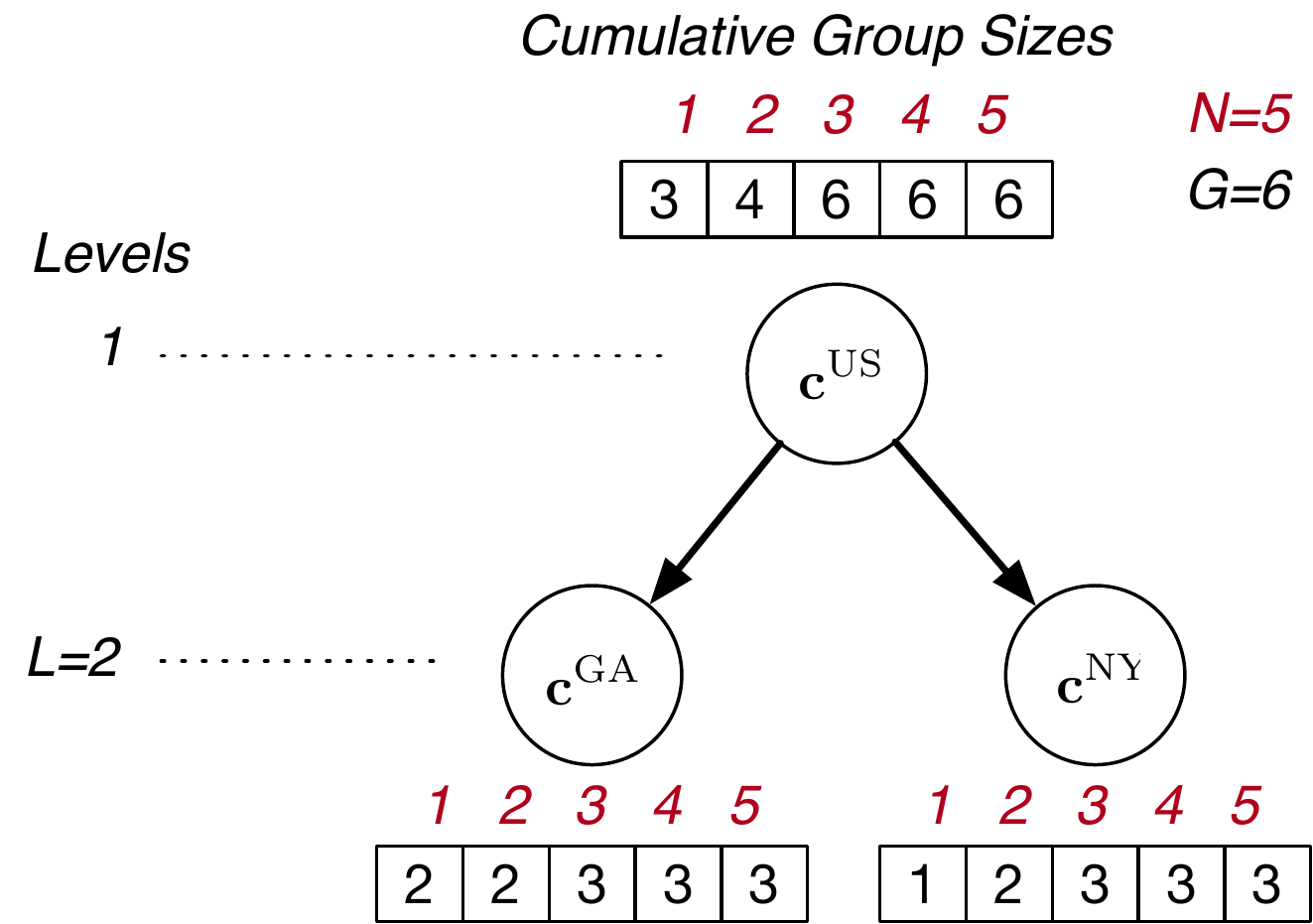}
	\caption{\label{fig:hierarhcy2} 
	\emph{Cumulative group size} hierarchy $\cT_2^c$ associated with the dataset of \Cref{tab:1}.
	}
\end{figure}

\noindent 
The result follows from the fact that removing an element from a group
$s$ in $\csizevec$ only decreases group $s$ by one and increases the
group $s-1$ preceding it by one. This idea is from \cite{hay:10},
where cumulative sizes are referred to as \emph{unattributed
histograms}.

\section{Cumulative PGSR Mechanisms}
\label{sec:cumulative_PGSR}

Operator $\oplus$ can be used to produce a hierarchy $\cT^c =
\{\csizevec^r | r \in \regionset\}$ of cumulative group sizes.  An
example of such a hierarchy is provided in \Cref{fig:hierarhcy2}.  To
generate a privacy-preserving version $\tilde{\cT}^c$ of $\cT^c$, it
suffices to apply the geometrical mechanism with parameter $\lambda =
(L / \epsilon)^N$ on the vectors $\csizevec^r$ associated with every
node $a_s^r$ for the groups of size $s$ in region $r$ % \nandoSide{check this notation, shouldn't this be $a$?}
of the region hierarchy. Once the noisy sizes are
computed, the noisy group sizes can be easily retrieved via an inverse
mapping $\ominus: \mathbb{Z}_+^N \to \mathbb{Z}_+^N$ from the
cumulative sums.

Note, however, that the resulting private versions $\tilde{\csizevec}^r$
of $\csizevec^r$ may no longer be non-decreasing (or even
non-negative) due to the added noise.  Therefore, as in Section
\ref{sec:IPPGSR}, a post-processing step is applied to restore
consistency and to guarantee the PGSR conditions \eqref{gse:2} to
\eqref{gse:4}. The post-processing is illustrated in Model \ref{fig:NPc}.  It
takes as input the noisy hierarchy of cumulative sizes $\tilde{\cT}^c$
computed with the geometrical mechanism and optimizes over variables
$\hat{\csizevec}^r =  (\hat{\csize}_1^r,\ldots, \hat{\csize}_N^r)$
for $r \in  \regionset$, minimizing the L2-norm with respect to their noisy
counterparts (\Cref{eq:C1}).  Constraints \eqref{eq:C2} guarantees that
the sum of the sizes equals the public value $G$ (PGSR
condition~\ref{gse:4}), where $\top$ denotes the root region of the
region hierarchy. Constraints \eqref{eq:C3} guarantee consistency of
the cumulative counts. Finally, Constraints \eqref{eq:C4} and
\eqref{eq:C5}, respectively, guarantee the PGSR consistency
(\ref{gse:2}) and validity (\ref{gse:3}) conditions.
% ------------------------------------  %
\begin{model}[!bt]
\caption{The $\cM_c$ post-processing step. \label{fig:NPc}}
\vspace{-8pt}
\begin{alignat}{2}
	\minimize_{ \left\{\hat{\csizevec}^r \right\}_{r \in \regionset}} \;\;
	&
	\sum_{r \in \regionset} \left\| \hat{\csizevec}^r - \tilde{\csizevec}^r \right\|_2^2 
	&&
	\tag{C1} \label{eq:C1}\\
	\subjectto\;\;
	&
	%%%%%%%%%
	\hat{\csize}^{\top}_N = G 
	&&
	\tag{C2} \label{eq:C2} \\
	%%%%%%%%
	&
	\hat{\csize}^r_i \leq \hat{\csize}^r_{i+1}
	&&\qquad
	\forall r \in \regionset, i \in [N-1] 
	\tag{C3} \label{eq:C3}\\
	%%%%%%%%
	&
	\sum_{r' \in \textsl{ch}(r)}	
	\hat{\csize}^{r'}_i = \hat{\csize}^r_{i}
	&&\qquad
	\forall r \in \regionset, i \in [N] 
	\tag{C4} \label{eq:C4}\\
	%%%%%%%%%
	&
	\hat{\csize}^r_i \in \{0, 1, \ldots\} 
	&&\qquad
	\forall r \in \regionset, i \in [N] 
	\tag{C5} \label{eq:C5}
\end{alignat}
\vspace{-16pt} 
\end{model} %
% ------------------------------------  %

Once the post-processed hierarchy $\hat{\cT}^c$ is obtained, the operator
$\ominus$ is applied to obtain a post-processed version of the group
size hierarchy $\hat{\cT}$. 
The resulting mechanism, denoted by $\cM_c$ is called the \emph{Cumulative PGSR Mechanism}. $\cM_c$ satisfies $\epsilon$-differential privacy, due to post-processing immunity,
similarly to the argument presented in \Cref{th:privacy_H} (see Appendix). This
mechanism is called the cumulative PGSR and denoted by $\cM_{c}$.
% \begin{wrapfigure}[11]{R}{0.46\textwidth}
	%\vspace{-20pt}
	% {\small
	\begin{algorithm}[!t]
		\caption{$\cM_c^{\textsl{dp}}$.post-process\hspace*{-25pt}}
		\label{alg:dp_calg}
	    \setcounter{AlgoLine}{0}
		\SetKwInOut{Input}{input}
        \SetKwInOut{Output}{output}
		\DontPrintSemicolon

        \Input{$\tilde{\cT}^c = \{ \tilde{\csizevec}^r \st r \in \regionset\}$}
        %\Output{$\hat{\cT}$}
        \ForEach{$r \in \regionset$}{
	    	$\displaystyle
	    	\hat{\csizevec}^r \gets 
	    	\hspace{16pt} 
	    	\argmin_{\hat{\csizevec}} \;\;
	    		\left\| \hat{\csizevec} - \tilde{\csizevec}^r \right\|^2$\;
	    	\nonl
	    	$\hspace{25pt} 
	    	\displaystyle
	    	\subjectto\;\; \hat{\csize}_i \leq \hat{\csize}_{i+1} \qquad\;\;\; \forall i \in [N-1]$\;
	    	\nonl
        	$\hspace{80pt} 
        	0 \leq \hat{\csize}_i \leq G \qquad \forall i \in [N]$\;
        	%%%%%%
        	$\bar{\gsizevec} \gets \ominus( \textsl{round}( \hat{\csizevec}^r) )$\;
        	$\hat{\cT} \gets \bar{\cT} \cup \{\bar{\gsizevec}\}$
        }
        $\hat{\cT} \gets \cM_{\textsl{dp}}.{\textsl{post-process}}(\hat{\cT})$
	\end{algorithm}
% 	}
% \end{wrapfigure}

\subsection{An Approximate Dynamic Programing Cumulative PGSR Mechanisms}
\label{sec:_apprdp_cumulative_PGSR}

Unlike for $\cM_{H}$, the structure of the PGSR problem cannot be
exploited directly to create a dynamic-programming mechanism. The constraints imposed by the optimization in Model~\ref{fig:NPc} do not allow for a hierarchical decomposition with recurrent
substructure as the inequalities \eqref{eq:C3} relate sibling nodes in the hierarchy.  Therefore, a dynamic program that exploits the
dependencies induced by the region hierarchy $\cT$, would need to
associate each node $\node^r$ in $\cT$ with a cost table $\ctable^r$
which may take as input vectors of size up to $N$. It follows that
each cost table $\ctable^r$ size is in $O(\binom{G+N-1}{N-1})$, which
makes this dynamic-programming approach computationally intractable,
especially for real-world applications.

To address this issue, a previous version of this work \cite{fioretto:CP-19}, proposed an approximated dynamic program version of the cumulative mechanism, called $\cM_c^{\textsl{dp}}$ that operates in three steps: 
\begin{enumerate}
\item It creates a noisy hierarchy $\tilde{\cT}^c$.
\item Next, it executes the post-processing step described in \Cref{alg:dp_calg}.
\item Finally, it runs the post-processing step of the polynomial-time PGSR mechanism (see \Cref{sec:dp_gse_exploit}).
\end{enumerate}

\noindent The important addition is in step 2 which takes $\tilde{\cT}^c$ as
input and, for each node $\tilde{\csizevec}^r$, solves the \emph{convex program} described in line 2 of \Cref{alg:dp_calg} to create a new noisy hierarchy
$\hat{\csizevec}^r$ that is non-decreasing and non-negative. The
resulting cumulative vector $\hat{\csizevec}^r$ is then rounded and
transformed to its corresponding group size vector through operator
$\ominus(\cdot)$ (line 3). The resulting vector $\bar{\gsizevec}^r$ is
added to the region hierarchy $\hat{\cT}$ (line 4). Observe that this
post-processing step pays a polynomial-time penalty with respect to the runtime of
the $\cM_H^{\textsl{dp}}$ post-processing. The convex program of line
(2) is executed in $O(\text{poly}(N))$ and the resulting
post-processing step runtime is in $O\big( |\regionset| \text{poly}(N)
+ |\regionset|N \bar{D} \log \bar{D} \big)$.

% \begin{figure}
% 	\centering\includegraphics[width=0.7\linewidth]{tree_3.pdf}
% 	\caption{\label{img:tree_regions2} Illustration of the region-tree for the DP hierarchical cumulative size histogram.}
% \end{figure}
% Figure \ref{img:tree_regions2} provides an illustration of the computations described above. 
% This approach, however, does not guarantee consistency on the less-than-or-equal relation between successive estimates for the cumulative group size within each level and each region. 

While, the experimental results (see \Cref{sec:experiments}) show 
that $\cM_c^{\textsl{dp}}$ consistently reduces the final error, 
when compared to $\cM_H^{\textsl{dp}}$, it is important to note that 
$\cM_c^{\textsl{dp}}$ does not solve the same post-processing program as the cumulative PGSR mechanism specified by \Cref{eq:C1,eq:C5,eq:C3,eq:C4,eq:C2}, since it restores consistency of the cumulative counts locally. 
To address this issue, this work introduces next an \emph{exact} DP cumulative PGSR mechanism.

\subsection{An Exact Dynamic Programming Cumulative PGSR Mechanisms}
\label{sec:dp_cumulative_PGSR}

% \begin{figure}[!t]
% 	\centering\includegraphics[width=\linewidth]{chain_hierarchy.pdf}\\
% 	\vspace{-0.2cm}
% 	{\footnotesize\hspace*{110pt}(a) \hspace{160pt}(b)}\\
% 	\caption{\label{img:c_tree} (a): Region hierarchy $\chaindp$ associated with the dataset of \Cref{fig:ex1}(a).
% 	(b): Example of cost table computation for subtree associated to the group 1 estimates.}
% \end{figure}

This section develops an exact and efficient dynamic programming
algorithm for solving the cumulative PGSR problem optimally.  It
relies on the observation that, when the cumulative operator $\oplus$
is applied to the hierarchy $\treedp$ instead than to the original
hierarchy $\cT$, the post-processing step of Model \ref{fig:NPc}
only imposes inequalities among sibling nodes within each
subtree, allowing the construction of overlapping sub-problems.

%%%%%%%%%%%%%%%%%%%%%%%%%%%%%%%%%%%%%%%
\begin{figure}[!t] \centering
\includegraphics[width=\linewidth]{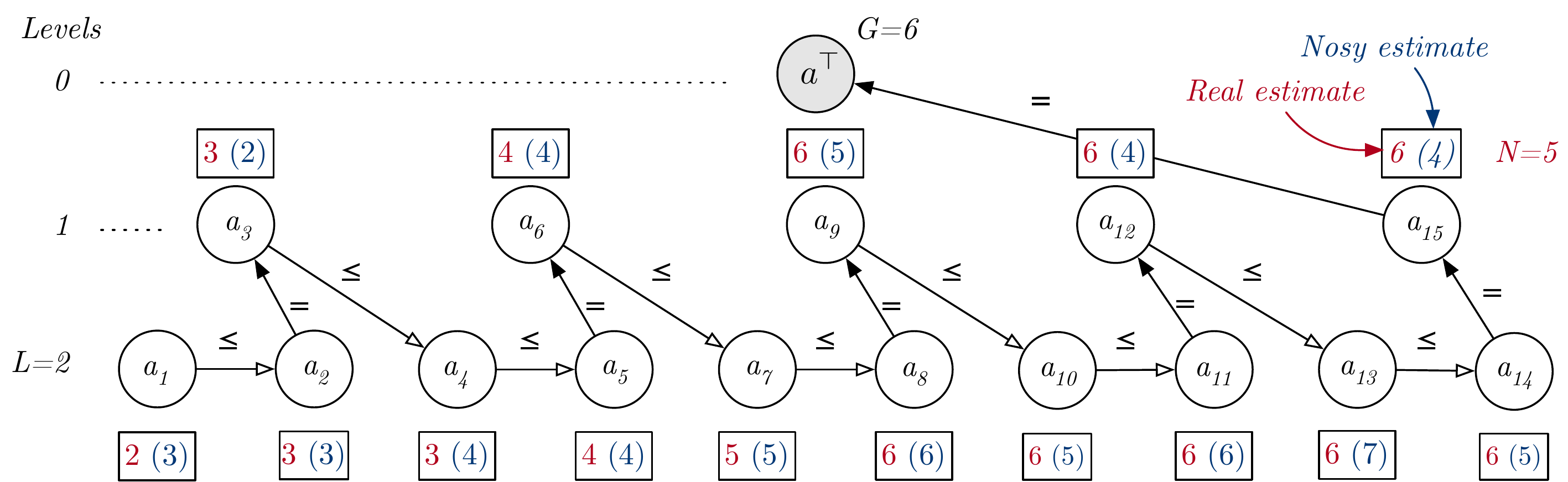}\\
\caption{\label{fig:c_tree} The $\chaindp$ node ordering associated
with the dataset of \Cref{tab:1}. %The gray notation indicates the
%corresponding nodes in the $\treedp_c$ hierarchy.
} 
\end{figure}
%%%%%%%%%%%%%%%%%%%%%%%%%%%%%%%%%%%%%%%

More precisely, the key insight is to build a chain of nodes by visiting the tree
hierarchy $\treedp$ (see \Cref{img:tree_regions}) using a post-order traversal. The resulting chain structure $\chaindp$ is composed of
$|\cR|\,N$ nodes, where $\node_i$ refers to the $i$-th node in the
post-order traversal of $\treedp$. This chain structure enables the
use of an efficient dynamic-programming algorithm to solve the
post-processing step over cumulative counts. The resulting mechanism
is referred to as the \emph{Chain-based Cumulative Dynamic Programming
PGSR} mechanism and denoted by $\cM_c^{\textsl{ch}}$. It operates in
three phases described as follows.

\subsubsection{Cumulative Hierarchy Pre-processing Phase}
Given a DP hierarchy $\treedp$, the algorithm constructs a chain
 $\tree^{\textsl{ch}} = \{c_i \mid i\in[|\cR|N]\}$
by traversing the nodes of 
$\treedp$ using a post-ordering scheme. Denote with $\node_i$ and $n_i$ the $i$-th node and group size value, respectively, in the post-order traversal of $\treedp$. 
The cumulative count $c_i$ associated with node $a_i$ is the sum of $n_i$ and all values $n_j$ associated with nodes $a_j$ that precedes $a_i$ and lie in the same level as $\node_i$ in $\treedp$. More formally, 
\[
    c_i = \sum_{j\in S_i} \gsize_j\,,\qquad S_i=\left\{j\mid j\leq i,~\text{lev}(a_j)=\text{lev}(a_i)\right\}\,,
\]
where %$\gsize_j$ represents the value associated with the node $\node_j$ in $\treedp$, and 
$\text{lev}(a)$ describes the level of node $a$ in $\treedp$.
An illustration of the chain structure $\cT^{ch}$ associated with the running example is shown in Figure \ref{fig:c_tree}. %where the values are outlined in red color while their noisy counterparts in blue.

% $\treedp_c = \{\bm{c}^\ell \mid \ell \in [L]\}$ where the 
% cumulative count vectors $\bm{c}^\ell$ include all values $n_r^s$, 
% associated with node $a_r^s$ in $\treedp$, such that 
% $\text{lev}(a_r^s) = \ell$. The elements in vector $\bm{c}^\ell$ are 
% ordered according to relation $\prec^\ell$ defined, for each level 
% $\ell \in [L]$ as:
% \[
% 	n_r^s \prec^\ell n_{r'}^{s'} \iff 
% 	\left(\text{lev}(a^s_r) = lev(a^{s'}_{r'})\right) 
% 	\land
% 	\left( (r = r' \land s < s')
% 	\lor
% 	(r \prec^\cR r')
% 	\right) 
% \]
% for some region ordering $\prec^\cR$ and where $\text{lev}(a)$ 
% describes the level associated to node $a$ in $\treedp$. 

\subsubsection{Privacy Phase}
Next, the algorithm constructs a noisy version of the $\tree^{\textsl{ch}}$ 
hierarchy, denoted $\tilde{\tree}^{\textsl{ch}}$ that is constructed by applying the geometrical mechanism with 
parameter $\lambda(L/\epsilon)^{|\cR_\ell|\, N}$ to the cumulative count 
vector $\bm{c}^\ell$ associated with the nodes at level $\ell$ for all
$\ell\in[L]$.
% This algorithm essentially creates a \emph{chain} structure in
% $\tree^{\textsl{ch}}$ from 
% the tree structure in $\treedp$
% by traversing the nodes of 
% $\treedp$ via a post-ordering scheme. 
Figure \ref{fig:c_tree} illustrates the resulting cumulative group values after the privacy-preserving phase in blue.
% The non-privacy-preserving counterpart of $\tilde{\cT}^{\textsl{ch}}$
% is denoted as $\cT^{\textsl{ch}}$.

% This algorithm then constructs a new representation for
% $\tilde{\cT}^{\textsl{dp}}_c$, called $\tilde{\cT}^{\textsl{ch}}$, 
% that forms a \emph{chain} structure by traversing the nodes of 
% $\tilde{\cT}^{\textsl{dp}}_c$ using a post-ordering scheme. 
% An illustration of the chain structure associated to the running 
% example is shown in Figure \ref{fig:c_tree}.
% The non-privacy-preserving counterpart of $\tilde{\cT}^{\textsl{ch}}$
% is denoted as $\cT^{\textsl{ch}}$.

% An example hierarchy $\tree^{\textsl{ch}}$ can be observed in \Cref{fig:c_tree}
% where the values, associated with the nodes, are outlined in red color.

% \subsubsection{Chain Construction and Privacy Phase}

\subsubsection{The Post-processing Phase}

Given a noisy version $\tilde{\bm{c}}^\ell$ of the cumulative group 
sizes, an application of the inverse mapping $\ominus$ on 
$\tilde{\bm{c}}^\ell$ for each level $\ell \in [L]$ is sufficient 
to retrieve the desired noisy group sizes. 
However, the resulting group sizes may not satisfy the PGSR validity 
and consistency conditions.  The third phase of $\cM_c$ is 
 a post-processing step to restore the PGSR conditions 
\ref{gse:2} to \ref{gse:4}.

% ------------------------------------  %
\begin{model}[!t]
\caption{The modified post-processing step. \label{fig:NPch}}
%\vspace{-8pt}
Let $\dot{N} = |\cR|\,N$ 
\begin{alignat}{2}
	\minimize_{ \{\hat{c}_i \}_{i\in [\dot{N}]} } \;\;
	&
	\sum_{i=1}^{\dot{N}} \| \hat{c}_i - \tilde{c}_i \|_2^2 
	&&
	\tag{C6} \label{eq:C6}\\ 
	\subjectto\;\;
	&
	%%%%%%%%%
	\hat{c}_{\dot{N}} = G
	&&
	\tag{C7} \label{eq:C7} \\
	%%%%%%%%
	&
	\hat{c}_i \leq \hat{c}_{i+1}
	&&\qquad
	\forall i \in \left\{ i \in [\dot{N}-1] \mid
				       \text{lev}(a_i)  = \text{lev}(a_{i+1}) \right\}
	\tag{C8} \label{eq:C8}\\
	%%%%%%%%
	&
	\hat{c}_i = \hat{c}_{i+1}
	&&\qquad
	\forall i \in \left\{ i \in [\dot{N}-1] \mid
				       \text{lev}(a_i) > \text{lev}(a_{i+1}) \right\}
	\tag{C9} \label{eq:C9}\\
	%%%%%%%%%
	&
	\hat{c}_i \in \{0, 1, \ldots\}
	&&\qquad
    \forall i \in [\dot{N}] 
	\tag{C10} \label{eq:C10}
\end{alignat}
\vspace{-16pt} 
\end{model} %
% ------------------------------------  %
This step is illustrated in Model \ref{fig:NPch}.   It aims at
generating a new chain $\hat{\cT}^{\textsl{ch}}$ whose counts are
close to their noisy counterparts \eqref{eq:C6} and satisfy
the faithfulness \eqref{eq:C7}, consistency \eqref{eq:C8} and 
\eqref{eq:C9}, and validity conditions \eqref{eq:C10}. Once the 
post-processed chain $\hat{\cT}^{\textsl{ch}}$ is obtained, the 
operator $\ominus$ is applied to obtain the corresponding version of 
the group size hierarchy $\hat{\cT}$. 

The advantage of this formulation is its ability to exploit the
chain structure $\chaindp$ associated with the region hierarchy $\treedp_c$,
via an efficient dynamic programming algorithm to solve the 
post-processing step.  
Similarly to $\cM_H^{\textsl{dp}}$, this dynamic program consists of 
two phases, a bottom-up and a top-down phase, corresponding to the 
direction in which the chain is traversed. 

In the bottom-up phase, each node $\node_i$ of $\chaindp$ constructs
a cost table $\ctable_i: \dom_i \to \mathbb{R}_+$, where 
$D_i \subseteq \NN$ is the domain associated with node $a_i$, mapping 
values to costs. Intuitively, the cost table $\ctable_i(v)$ represents 
the optimal cost for the first $i$ nodes when the post-processed 
cumulative group size of the node $\node_i$ is equal to $v$. The cost 
tables for all nodes are updated starting from the head of the chain 
to its tail, according to the following recurrence relationship:
\begin{subequations}
\label{fig:ch_bottom_up}
\begin{alignat}{3}
  	\ctable_{i+1}(v) = 
		&\left(v - \tilde{c}_{i+1} \right)^2
		&& + 
		&& %\tag{d1} 
		\label{eq:d5}\\
		%%%%%%%%%%%%%
		& \hspace{20pt} \ctablech_{i+1}(v) 				   
		&&= \begin{cases}
		    \ctable_i(v) & \text{if } \mathrm{lev}(a_i) > \text{lev}(a_{i+1}),\\
		    \underset{ \substack{x_i\in\dom_i\\ x_i\leq v} }{ \minimize } \ctable_i(x_i) & \mathrm{otherwise.}
		\end{cases}
		%\tag{d2} 
		\label{eq:d6} 
\end{alignat}
\end{subequations}
The formulation for the cost $\ctable_{i+1}(v)$ has two components. 
The first component captures the deviation of the post-processed value 
$v$ from the noisy cumulative group size $\tilde{\csize}_{i+1}$ 
(\Cref{eq:d5}). The second component captures the optimal post-processing cost 
associated with the first $i$ nodes of the chain, provided that the 
succeeding node $\node_{i+1}$ takes on value of $v$ (\Cref{eq:d6}).

In the top-down phase, the algorithm traverses each node from the tail
of the chain to its head. 
The node $\hat{\csize}_{\vert\regionset\vert N}$ is set to value
$G$ (to satisfy Constraint \eqref{eq:C7}). Each visited node 
$\node_{i+1}$, for any $i\in[\vert\regionset\vert N-1]$, receives the 
post-processed cumulative group size $\hat{\csize}_{i+1}$ and determines 
the optimal post-processed solution for its predecessor $\node_i$ 
by solving $\ctablech_{i+1}(\hat{\csize}_{i+1})$ given in
\eqref{eq:d6}.

\begin{example}
An example of the mechanism $\cM_c^{\textsl{ch}}$ executed on a few
steps of the cost tables  $\ctable_1^{\GA},
\ctable_1^{\MI},\ctable_1^{\US}$ is shown in \Cref{tab:Tch}.
Consider the computation of the
%%%%%%%%%%%%%%%%%%%%%%%%%%%%%%%%%
\parbox[t]{\textwidth}%{\dimexpr\textwidth-\leftmargin} 
{%
\vspace{-2.5mm}
\begin{wrapfigure}[10]{R}{0.5\textwidth}%
% \vspace{-12pt}
\centering
\resizebox{0.99\linewidth}{!}
{%
\begin{tabular}{@{}r |r r@{}l@{}}
	\toprule
	$\bm{v}$ & $\tau_1^{\GA}$ & $\tau_1^\MI$ & \\[-1pt]
	\midrule
	$\bm{0}$ & 9       & $9\!+\!9$ &     $=\!\min(9)$\\[-1pt]
	$\bm{1}$ & 4  	   & $4\!+\!4$ &     $=\!\min(9,4)$\\[-1pt]
	$\bm{2}$ & 1  	   & $1\!+\!1$ &     $=\!\min(0,4,1)$\\[-1pt]
	$\bm{3}$ & $^{*}0$ & $^{*}0\!+\!0$ & $=\!\min(9,4,1,0)$\\[-1pt]
	$\bm{4}$ & 1       & $1\!+\!0$ &     $=\!\min(9,4,1,0,1)$\\[-1pt]
	\bottomrule
\end{tabular}
~~
\begin{tabular}{@{}r |r r@{}}
	\toprule
	$\bm{v}$ & $\tau_1^{\MI}$ & $\tau_1^\US$ \\[-1pt]
	\midrule
	$\bm{0}$ & 18      & $4 \!+\! 18$ 		\\[-1pt]
	$\bm{1}$ & 8       & $1 \!+\! 8$ 		\\[-1pt]
	$\bm{2}$ & 2       & $0 \!+\! 2$ 		\\[-1pt]
	$\bm{3}$ & $^{*}$0 & $^{*}1 \!+\! 0$ 	\\[-1pt]
	$\bm{4}$ & 1       & $4 \!+\! 1$ 		\\[-1pt]
	\bottomrule
\end{tabular}
}%
\captionof{table}{\label{tab:Tch}
Example of cost table computation for subtree associated to the group  1 estimates.}
\end{wrapfigure}
%%%%%%%%%%%%%%%%%%%%%%%%%%%%%%%%%
cost function $\ctable_1^{\US}$ in the bottom-up phase. 
To start with, the algorithm initializes the cost
table $\ctable_1^{\GA}$ as $\ctable_1^{\GA} (v)=\left\vert
v-3\right\vert^2$, where the noisy cumulative group size associated
with  the node $\node_1^\GA$ is $3$. Notice that the post-processed
cumulative group size of this node $\node_1^\GA$ is not supposed to
exceed that of its succeeding node $\node_1^\MI$. 
Thus, the algorithm updates the cost table $\ctable_1^\MI$ by aggregating the following
two parts: 
The one including the cost associated with the current node
$\left\vert v-3\right\vert^2$ and the that including the optimal} cost for its preceding node(s) $\min_{0\leq x\leq v} \ctable_1^\GA(x)$. 
It follows that $\ctable_1^\MI(3)=\left\vert 3-3\right\vert^2+\min_{0\leq x\leq 3}
\ctable_1^\GA(x)=0+\ctable_1^\GA(3)=0+0=0$. Then, for the node
$\node_1^\US$, its post-processed cumulative group size equals that of
its preceding node $\node_1^\MI$. As a result, its associated cost,
say $\ctable_1^\US(3)$, is composed by $\left\vert 3-2\right\vert^2=1$
and $\ctable_1^\MI(3)=0$. The table marks the value selected during the 
top-down phase with a $^{*}$ symbol.
\end{example}

The next results discuss the piecewise linear convexity of the cost function $\ctablech$ and the computational complexity of the algorithm. 
% We present in \Cref{lm:ch_convexity} that the cost tables turn out to be CPWL.
% This observation leads us to the efficiency result given by \Cref{tm:efficiency}.
% The proofs of the following results are deferred to Appendix. 
\begin{lemma}
\label{lm:ch_convexity}
The cost table $\ctable_i$ of each node $\node_i$ of $\chaindp$ is CPWL.
\end{lemma}

\begin{theorem}
\label{tm:efficiency}
The cost table $\ctable_i$ for each $i\in[\vert\regionset\vert N]$ can be computed 
in time $O\big(\bar{D} \big)$, where $\bar{D} = \max_{i} |D_i|$ for
$i\in[\vert\regionset\vert N]$.
\end{theorem}

\section{Related Work}
\label{sec:related_work}

The release of privacy-preserving datasets using differential privacy
has been subject of extensive research \cite{huang:15,li:09,li:14}. 
These methods focus on creating \emph{unattributed histograms} that 
count the number of individuals associated with each possible property 
in the dataset universe. During the years, more sophisticated 
algorithms have been proposed, including those exploiting the problem 
structure using optimization to improve accuracy \cite{cormode:12,li:14,qardaji:13,li2010optimizing}.

Additional extensions to consider hierarchical problems were also explored. 
Hay et al.~\cite{hay:10} and Qardaji et al.~\cite{qardaji:13} study 
methods to answer count queries over ranges using a hierarchical 
structure to impose consistency of counts. 
While hierarchies contribute an additional level of fidelity for 
realizing a realistic data release it further challenge the privacy/accuracy 
tradeoff. 
Other methods have also incorporated partitioning scheme to the data-release
problem to further increase the accuracy of the privacy-preserving 
data by cleverly splitting the privacy budget in different hierarchical 
levels \cite{xiao2010differentially,cormode2012differentially,zhang2016privtree}.

These methods differ in two ways from the mechanisms proposed here:
(1) They focus on histograms queries, rather than group queries; the
latter generally have higher $L_1$-sensitivity and thus require more
noise and (2) they ensure neither the consistency for integral counts
nor the non-negativity of the release counts. They thus violate the
requirements of group sizes (see \Cref{sec:group_est}).

Fioretto and Van Hentenryck recently proposed a hierarchical-based
solution based on minimizing the L2-distance between the noisy counts
and their private counterparts \cite{fioretto:AAMAS-18}. While this 
solution guarantees non-negativity of the counts, their mechanism, 
if formulated as a MIP/QIP, cannot cope with the scale of the census 
problems discussed here which compute privacy-preserving country-wise 
group sizes. If their solution is used as is, in its relaxed form, 
then it cannot guarantee the integrality of the counts. These mechanisms 
reduce to $\cM_H^r$, which, has is shown to be strongly dominated by 
the DP-based mechanisms (see Section \ref{sec:experiments}).

A line of work that is closely related to the problem analyzed in this 
paper is that initiated by Blocki et al.~\cite{blocki2016differentially} 
and Hay et al.~\cite{hay:10} that study \emph{unattributed histogram} 
which are often used to study node degrees in a path in graphs.  
Unattributed histograms are used to answer queries of the type: 
``How many people belong to the $k$-th largest group?''. They can 
be far more accurate than naively adding noise to each group and 
then selecting the $k$-th largest noisy group~\cite{hay:10}. They are 
however different from the queries used in this work as they count 
people rather then groups and are not necessarily hierarchical. 
A substantial contribution is represented by the work of Kui et 
al.~\cite{kuo2018_VLDB}, that study the problem of releasing 
hierarchical group queries that satisfies the non-negativity, 
integrality, and consistency of the counts. They propose a 
solution that, similarly to \cite{fioretto:AAMAS-18}, can be 
mapped to $\cM_H^r$ and apply rounding to ensure integrality, 
as well as studying the context of cumulative group queries. 

Finally, an important deployment is represented by the TopDown
algorithm \cite{abowd2018us}, used by the US Census to for the 2018
end-to-end test, in preparation for the 2020 release. The algorithm is
a weaker form of $\cM_H^r$ that first obtains inexact noisy counts to
satisfy the desired privacy level and then alternates the following
two steps for multiple levels of a geographic hierarchy, from top to
bottom, as the name suggests. The first step is a program analogous to $\cM_H^r$, but it operates on two consecutive levels only. The
second step is an ad-hoc rounding strategy to guarantee the
integrality of the counts while satisfying the hierarchical
invariants. These two steps are repeated processing two contiguous
levels of the hierarchy until the leaves are reached.

\section{Experimental Evaluation}
\label{sec:experiments}

This section evaluates the proposed privacy-preserving mechanisms for
the PGSR problem. The evaluation focuses on comparing runtime and
accuracy of the mechanisms described in the paper.   Consistent with
the privacy literature, accuracy is measured in term  of the $L_1$
difference between the privacy-preserving group sizes  and the
original ones, i.e., given the original group sizes $\cT =
\{\gsizevec^r \st r \in \regionset\}$, and their private counterparts
$\hat{\cT}$, the $L_1$-error is defined as $\sum_r \| \gsizevec^r -
\hat{\gsizevec}^r\|_1$.  Since the mechanisms are non-deterministic
due to the noise added by the geometric mechanism, $30$ instances are
generated for each benchmark and the results report average values and
standard deviations. Each mechanism is run on a single-core $2.1$ GHz
terminal with $24$GB of RAM and is implemented in Python $3$ with
Gurobi $8.1$ for solving the convex quadratic optimization problems.

\paragraph{\textbf{Mechanisms}} 
The evaluation compares the PGSR mechanisms $\cM_{H}$, its cumulative
version $\cM_c$, and their polynomial-time dynamic-programming (DP)
counterparts $\cM_H^{\textsl{dp}}$, $\cM_c^{\textsl{dp}}$, and
$\cM_c^{\textsl{ch}}$.  The former are referred to as OP-based methods
and the latter as DP-based methods.  In addition to $\cM_{H}$ and
$\cM_c$, that solve the associated post-processing QIPs, the
experiments evaluate the associated \emph{relaxations}, $\cM_{H}^r$
and $\cM_c^r$, respectively, that relax the integrality constraints
\eqref{eq:H4} and \eqref{eq:C5} and rounds the solutions.  %, while
still guaranteeing non-negativity and the final solutions are rounded.
  For completeness, the experiments also evaluate the performance of
the TopDown algorithm \cite{abowd2018us} and  the optimization-based
mechanism $\cM_{\textsl{dp}}^{\text{OP}}$ that \emph{does not} exploit
the structure of the cost function to compute the cost tables.

\begin{figure}[!t]
\centering
\includegraphics[width=0.65\textwidth]{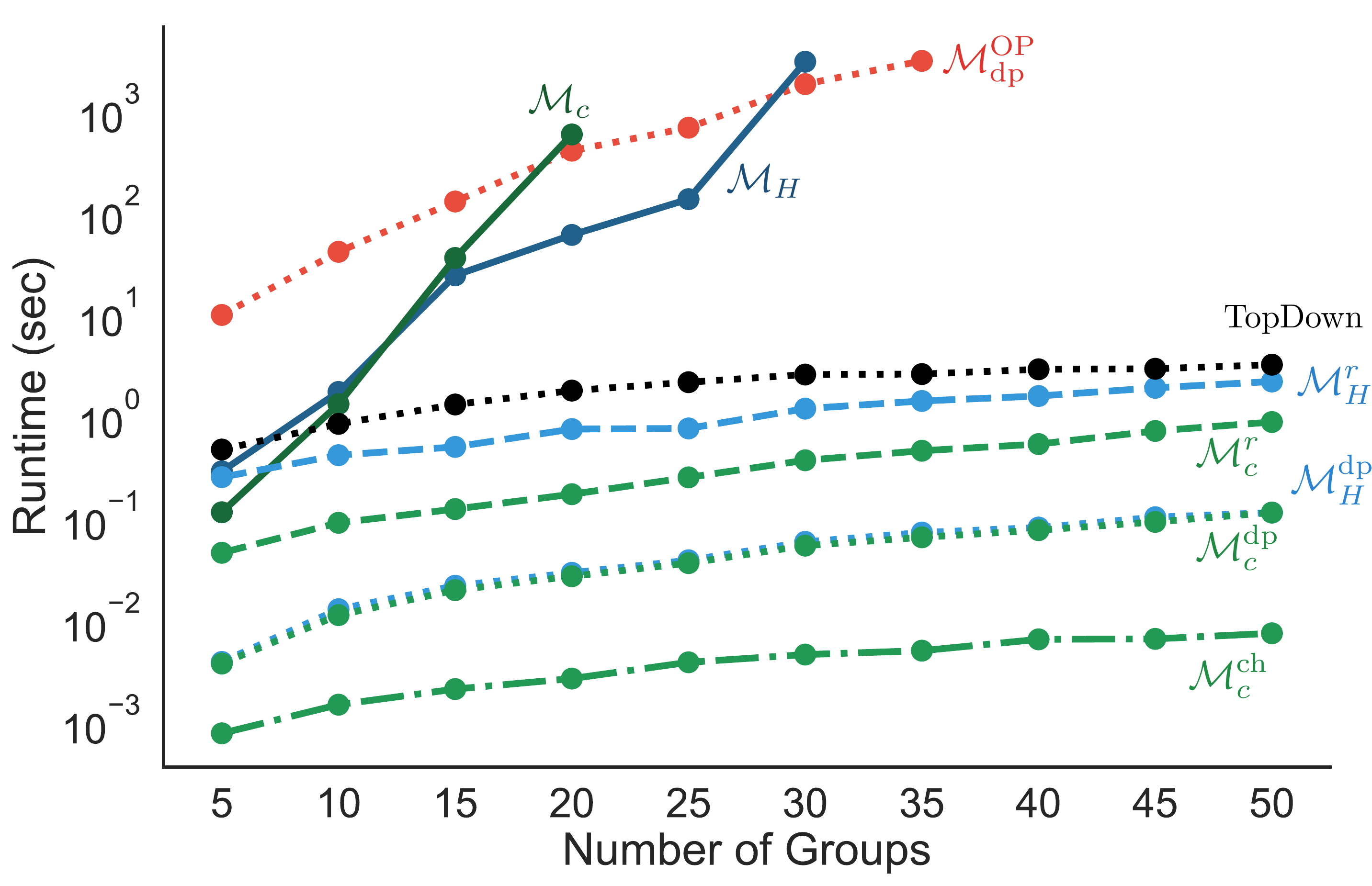}
\caption{\label{fig:runtime} Runtime (in seconds) at
varying of the number of group size $N$.}
\end{figure}

\paragraph{\textbf{Datasets}} The mechanisms are evaluated on three datasets.

\begin{itemize}
%\noindent$\bullet~$ 
\item \textbf{Census Dataset}: The first dataset has
117,630,445 groups, 7,592 leaves, 305,276,358 individuals, 3
levels, and $N$=1,000. Individuals live in facilities, i.e., households
or dormitories, assisted living
facilities, and correctional institutions. Due to privacy concerns and
lack of available methods to protect group sizes during the 2010
Decennial Census release, group sizes were aggregated for any
facility of size 8 or more (see Summary File 1 \cite{census-files}).
Therefore, following \cite{kuo:18} and starting from the truncated
group sizes Census dataset, the experiments augment the dataset with
group sizes up to $N$=1,000 that mimic the published statistics, but
add a heavy tail to model group quarters (dormitories, correctional
facilities, etc.). This was obtained by computing the ratio $r =
\gsize_7 / \gsize_6$ of household groups of sizes 7 and 6, subtracting
from the aggregated groups $\gsize_{8+}$ $M$ people according to the
ratio $r$, and redistributed these $M$ people in groups $k > 8$ so
that the ratio between any two consecutive groups holds (in
expectation).  Finally, $50$ outliers were added, chosen uniformly in the interval between $10$ and $1,000$. The region hierarchy is composed by
the National level, the State levels (50 states + Puerto Rico and
District of Columbia), and the Counties levels (3144 in total).

%\noindent$\bullet~$ 
\item \textbf{NY Taxi Dataset}: The second dataset has 13,282 groups,
  3,973 leaves, 24,489,743 individuals, 3 levels, and $N$=13,282. The
  2014 NY city Taxi dataset \cite{NY-taxi-data} describes trips
  (pickups and dropoffs) from geographical locations in NY city. The
  dataset views each taxi as a group and the size of the group is the
  number of pickups of the taxi. The region hierarchy has 3 levels:
  the entire NY city at level 1, the boroughs: \emph{Bronx},
  \emph{Brooklyn}, \emph{EWR}, \emph{Manhattan}, \emph{Queens}, and
  \emph{Staten Island} at level 2, and a total of 263 zones at level
  3.

%\noindent$\bullet~$ 
\item \textbf{Synthetic Dataset}: Finally, to test the
 runtime scalability, the experiments considered synthetic data from
 the NY Taxi dataset by limiting the number of group sizes $N$
 arbitrarily, i.e., removing group sizes greater than a certain
 threshold.
\end{itemize}
 
\subsection{Scalability}

The first results concern the scalability of the mechanisms, which are
evaluated on the synthetic datasets for various numbers of group
sizes.  Figure \ref{fig:runtime} illustrates the runtimes of the
algorithms at varying of the number of group sizes $N$ from $5$ to
$50$ for the synthetic dataset. The experiments 
have a timeout of $30$ minutes and the runtime is reported in log-$10$
scale.  The figure shows that the exact OP-based approaches and
$\cM_{\textsl{dp}}^{\text{OP}}$ are not competitive, even for small
groups sizes. Therefore, these results rule out the following
mechanisms: $\cM_{\textsl{dp}}^{\text{OP}}$, $\cM_{H}$, and $\cM_{c}$
and the remaining results focus on comparing the relaxed versions of
the OP-based mechanisms versus their proposed DP-counterparts.

\subsection{Runtimes}

\begin{figure}[!t]
	\centering
	\includegraphics[width=0.7\linewidth]{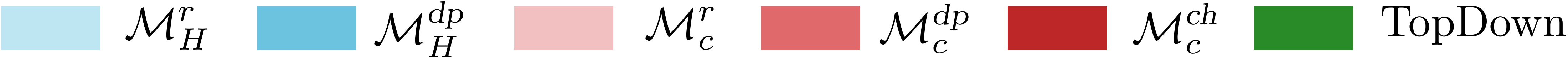}\\
	\includegraphics[width=0.45\linewidth]{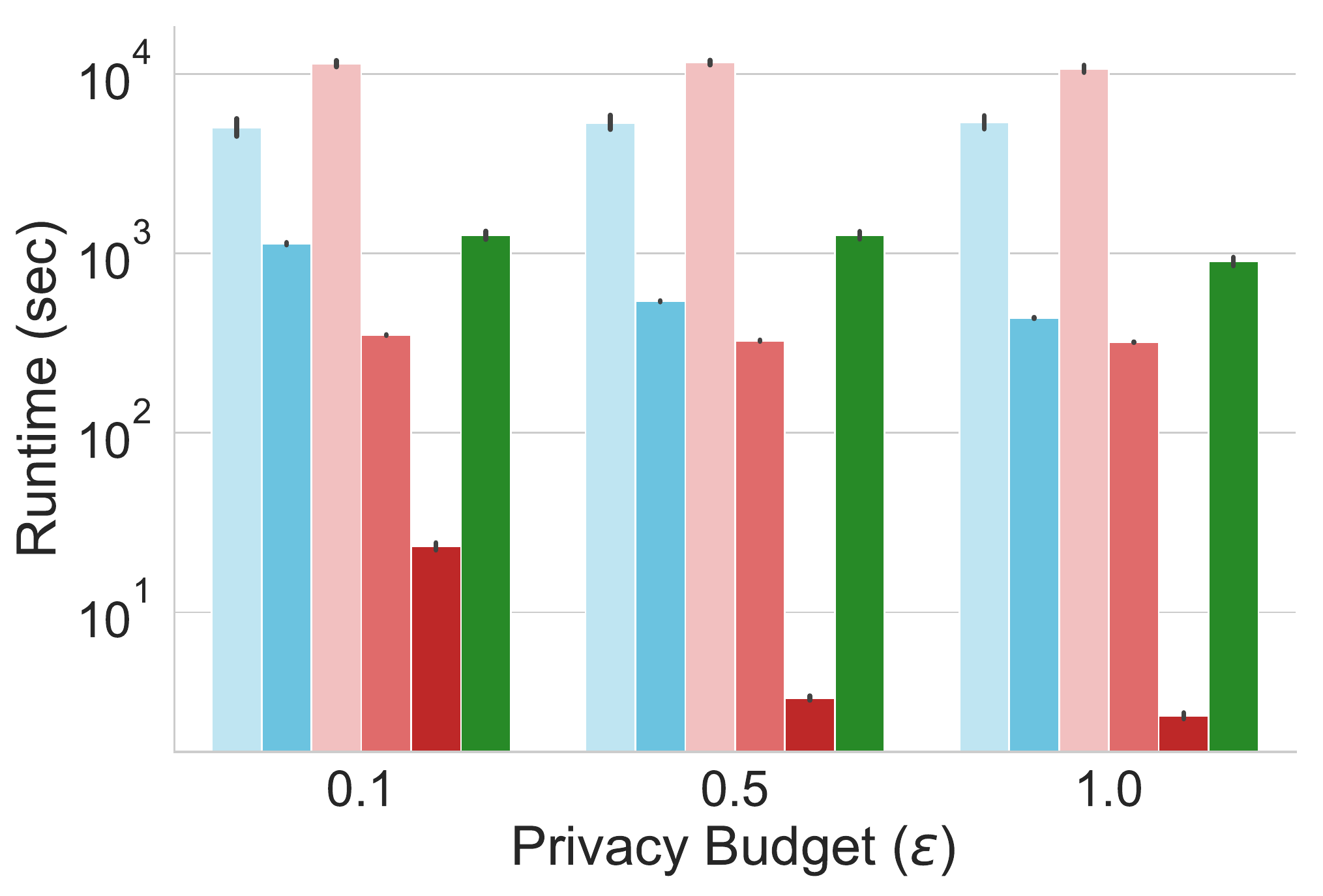}
	~~
	\includegraphics[width=0.45\linewidth]{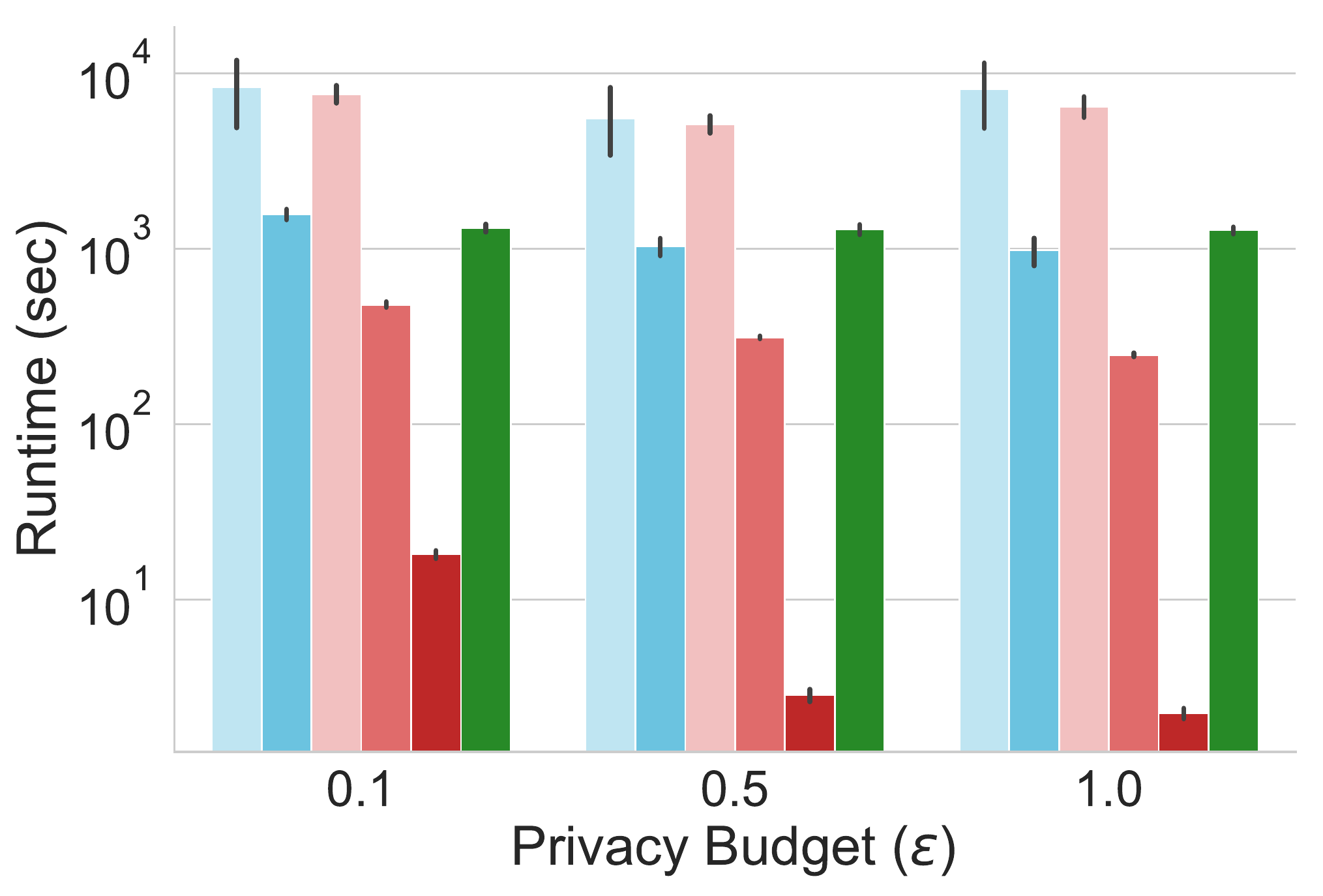}\\
	\caption{\label{fig:run} 
	The Runtime for the mechanisms: Census data (left) and Taxi data (right).}
\end{figure}

\Cref{fig:run} reports the runtime, in seconds, for the hierarchical
mechanism $\cM_{H}^r$ and its DP-counterpart $\cM_H^{\textsl{dp}}$, and the
hierarchical cumulative mechanism $ \cM_c^r$ and its DP-counterpart
$\cM_c^{\textsl{dp}}$, together with $\cM_c^{\textsl{ch}}$ and TopDown.
The left figure illustrates the results for the Census data and the right one for the NY Taxi data. The main observations can be summarized as follows:

\begin{enumerate}%$[leftmargin=*, parsep=0pt, itemsep=0pt, topsep=2pt]
\item Although the OP-based algorithms consider only a relaxation of
the problem, the \emph{exact} DP-versions are consistently faster. In
particular, $\cM_H^{\textsl{dp}}$ is up to one order of magnitude
faster than its counterpart $\cM^r$, and $\cM_c^{\textsl{dp}}$ is up
to two orders of magnitude faster than its counterpart $\cM_c^r$.

\item The proposed DP-based mechanisms are always faster than the TopDown algorithm and the newly proposed proposed $\cM_c^r$ mechanism is up to two order of magnitude faster than TopDown.

\item $\cM_c^r$ is consistently slower then $\cM_{H}^r$. This is because,
despite the fact that the two post-processing steps have the same
number of variables, the $\cM_c$ post-processing step has 
many additional constraints of type \eqref{eq:C3}.

\item The runtime of the DP-based mechanisms decreases as the privacy
budget increases, due to the sizes of the cost tables that depend on
the noise variance.\footnote{The implementation uses $\dom^r_s \!=\!
\{\tilde{\gsize}^r_s \!-\! \delta \ldots \tilde{\gsize}^r_s \!+\! \delta
\} \cap \mathbb{Z}_+$, where $\delta = 3 \times
\lceil2\lambda^2\rceil$, i.e., $3$ times the variance associated with
the double-geometrical distribution with parameter $\lambda$.}

\item The cumulative version $\cM_c^{\textsl{dp}}$ outperforms its
$\cM_H^{\textsl{dp}}$ counterpart. Once again, the reason is due to
the domain sizes. In fact, due to reduced sensitivity,
$\cM_c^{\textsl{dp}}$ applies a smaller amount of noise than that required by $\cM_H^{\textsl{dp}}$ to guarantee the same level of privacy and resulting in smaller domain sizes.

\item $\cM_c^{\textsl{ch}}$ is consistently faster than $\cM_c^{\textsl{dp}}$, which results from that computing the cost tables of the former is achieved with fewer operations than for the latter, as analyzed in \Cref{tm:dp_efficiency} and \Cref{tm:efficiency}.
\end{enumerate}

\begin{figure}[!t]
	\centering
	\includegraphics[width=0.7\linewidth]{AIJ_legend.pdf}\\
	\includegraphics[width=0.45\linewidth]{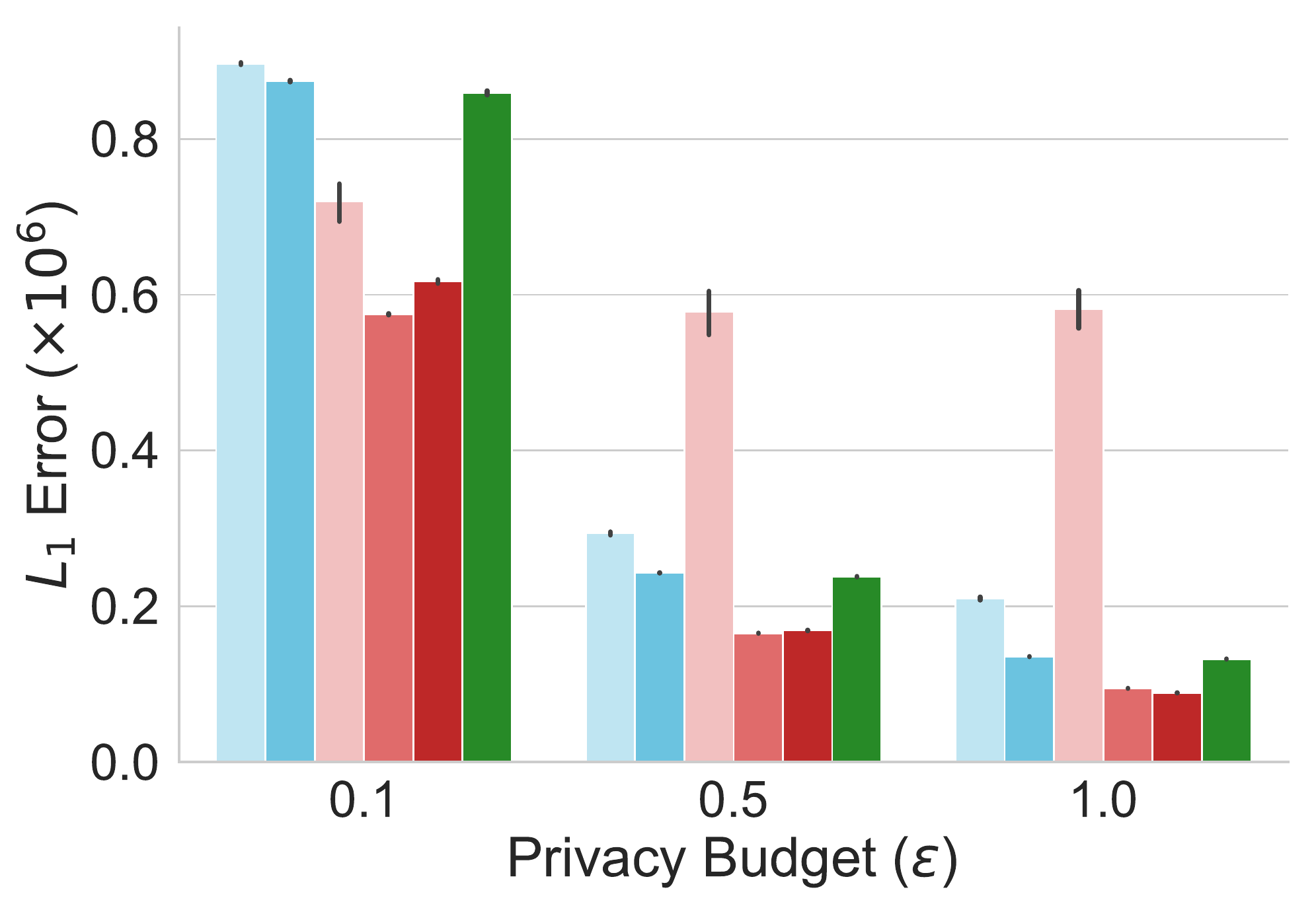}
	~~
	\includegraphics[width=0.45\linewidth]{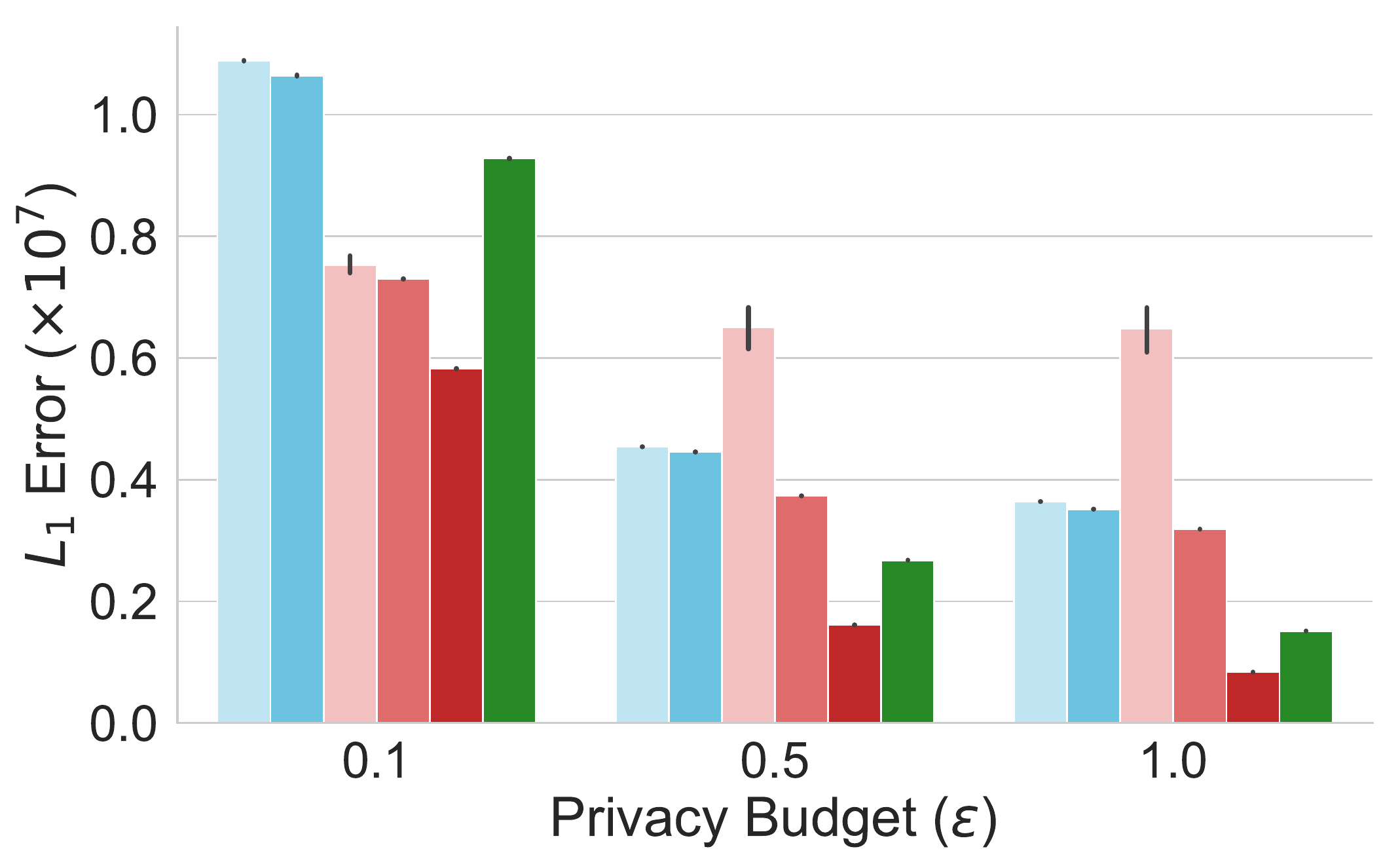}
	\caption{\label{fig:err} 
	The $L_1$ errors for the algorithms: Census data (left) and Taxi data (right). }
\end{figure}

\subsection{Accuracy}

\Cref{fig:err} reports the error induced by the
mechanisms, i.e., the $L_1$-distance between the
privacy-preserving and original datasets.
The main observations can be summarized as follows:

\begin{enumerate}%[leftmargin=*, parsep=0pt, itemsep=0pt, topsep=2pt]
%[topsep=0ex,itemsep=0ex,partopsep=0ex,parsep=0ex]
\item The DP-based mechanisms produce more accurate results than their
counterparts and $\cM^{\textsl{ch}}_c$ dominates all other mechanisms, including TopDown.

\item As expected, the error of all mechanisms decreases as the
  privacy budget increases, since the noise decreases as privacy
  budget increases.  The errors are larger in the NY Taxi dataset, which
  has a larger number of group sizes than the Census dataset.

\item Finally, the results show that the cumulative mechanisms tend to
concentrate the errors on small group sizes. Unfortunately, these are
also the most populated groups, and this is true for each
subregion of the hierarchy. On the other hand, the DP-based version,
that retains the integrality constraints, better redistributes the
noise introduced by the geometrical mechanism and produce 
substantially more accurate results.
\end{enumerate}

\noindent
To shed further light on accuracy, \Cref{tab:results} reports a
breakdown of the average errors of each mechanism at each level of the
hierarchies. Mechanism $\cM^{\textsl{ch}}_c$ is clearly the most
accurate. Note that the table reports the average number of
\emph{constraint violations} in the output datasets. A constraint
violation is counted whenever a subtree of the hierarchy violates the
PGSR \emph{consistency} condition \eqref{gse:2}.  Being exact, the
DP-based methods report no violations. In contrast, both $\cM_H^r$ and $\cM_c^r$ report a substantial amount of constraint violations.

\begin{table}[!b]
\centering
\resizebox{0.8\linewidth}{!} 
{%%
\begin{tabular}{ll|rrrr | rrrr}
    & \multicolumn{5}{c}{\bb{Taxi Data}} 
    & \multicolumn{4}{c}{\bb{Census Data}} \\

\toprule
    \multicolumn{2}{c}{}    
    	   	&  \multicolumn{3}{c}{$L_1$ Errors $(\times 10^4)$} &  ~~\#CV~ 
				& 
				\multicolumn{3}{c}{$L_1$ Errors $(\times 10^3)$} &  ~~~\#CV\\
    $\epsilon$ & 
    \multicolumn{1}{l}{\textbf{Alg}} & ~Lev 1 & ~~Lev 2 & ~~~Lev 3~&
    \multicolumn{1}{c}{}	  &~~~~Lev 1 & ~~Lev 2 & ~~~Lev 3~&   \\
\midrule
\multirow{6}{*}{\textbf{0.1}} 
	& $\cM_H^r$ 				&  25.4 &  158.7 &  904.4 &   18206 
							&  40.3 &  54.3 &  802.1 &    1966 \\
	& $\cM_H^{\textsl{dp}}$ 	&  26.6 &  121.9 &  915.7 &       \bb{0}
							&  10.3 &  38.4 &  825.4 &       \bb{0} \\
	& $\cM_c^r$ 				&  47.9 &  153.2 &  551.6 &   19460 
							&  23.1 &  64.5 &  632.2 &    1715 \\
    & $\cM_c^{\textsl{dp}}$ &  19.9 &   65.6 &  644.3 &       \bb{0} 
    						&   0.9 &  23.2 &  \bb{550.6} &       \bb{0} \\
    & $\cM_c^{\textsl{ch}}$ &  \bb{5.5} &   \bb{39.3} &  \bb{537.8} &       \bb{0} 
    						&   \bb{0.2} &  \bb{13.9} &  603.0 &       \bb{0} \\
    & TopDown				&  26.6 &  93.0 &  809.1 &   \bb{0} 
							&  3.7 &  35.0 &  820.5 &   \bb{0}\\
    \cline{1-10}
    \multirow{6}{*}{\textbf{0.5}} 
    & $\cM_H^r$ 				&   8.6 &   81.2 &  364.2 &   18591 
    						&  39.4 &  37.9 &  216.3 &    1990 \\
    & $\cM_H^{\textsl{dp}}$ 	&   5.5 &   31.0 &  408.9 &       \bb{0} 
    						&   2.4 &   9.4 &  230.8 &       \bb{0} \\
    & $\cM_c^r$ 				&  46.7 &  153.5 &  450.7 &   19531 
    						&  23.1 &  61.0 &  494.2 &    1718 \\
    & $\cM_c^{\textsl{dp}}$ &   4.0 &   16.4 &  352.9 &       \bb{0} 
    						&   0.2 &   5.8 &  \bb{159.1} &       \bb{0} \\
    & $\cM_c^{\textsl{ch}}$ &   \bb{1.2} &   \bb{9.5} &  \bb{150.6} &       \bb{0} 
    						&   \bb{0.0} &   \bb{3.4} &  165.3 &       \bb{0} \\
    & TopDown				&  5.3 &  20.3 &  242.1 &   \bb{0} 
							&  0.9 &  8.6 &  228.4 &   \bb{0} \\
    \cline{1-10}
    \multirow{6}{*}{\textbf{1.0}} 
    & $\cM_H^r$ 				&   7.7 &   77.2 &  279.0 &   18085 
    						&  40.7 &  39.2 &  130.0 &    1989 \\
    & $\cM_H^{\textsl{dp}}$ 	&   3.1 &   19.8 &  328.5 &   \bb{0} 
    						&   1.2 &   5.1 &  128.8 &    \bb{0} \\
    & $\cM_c^r$ 				&  47.1 &  154.2 &  447.1 &   19706 
    						&  24.1 &  63.0 &  494.5 &    1728 \\
    & $\cM_c^{\textsl{dp}}$ &   2.0 &    8.7 &  307.8 &       \bb{0} 
    						&   0.1 &   3.2 &   91.0 &       \bb{0} \\
    & $\cM_c^{\textsl{ch}}$ &   \bb{0.5} &    \bb{4.5} &  \bb{78.6} &       \bb{0} 
    						&   \bb{0.0} &   \bb{1.7} &   \bb{86.9} &       \bb{0} \\
    & TopDown				&  2.6 &  10.3 &  138.4 &   \bb{0}
							&  0.5 &  4.6 &  127.2 &   \bb{0} \\
    \bottomrule
\end{tabular}
}
\caption{\label{tab:results} $L_1$-errors and constraint violations (CV) for each level of the hierarchies.}
\end{table}

\section{Conclusions}
\label{sec:conclusions}

The release of datasets containing sensitive information concerning a
large number of individuals is central to a number of statistical
analysis and machine learning tasks. Of particular interest are
hierarchical datasets, in which counts of individuals satisfying a
given property need to be released at different granularities (e.g.,
the location of a household at a national, state, and county levels).
The paper discussed the \emph{Privacy-preserving Group Release} (PGRP)
problem and proposed an exact and efficient constrained-based approach
to privately generate consistent counts across all levels of the
hierarchy.  This novel approach was evaluated on large, real datasets
and results in speedups of up to two orders of magnitude, as well as
significant improvements in terms of accuracy with respect to
state-of-the-art techniques.  Interesting avenues of future directions
include exploiting different forms of parallelism to speed up the
computations of the dynamic programming-based mechanisms even further,
using, for instance, Graphical Processing Units as proposed
in~\cite{fioretto:Constraint-18}.

% \noindent\textbf{Reproducibility}: Datasets and model used will be released upon publication.

% \section*{References}
\bibliographystyle{plain}
\bibliography{cp19}
\appendix

\setcounter{theorem}{0}
\setcounter{lemma}{3}
\setcounter{AlgoLine}{0}

\section{Missing Proofs}
\label{sec:missing_proofs}

\begin{lemma} 
	The sensitivity $\Delta_{\gsizevec}$ of the group estimate
	query is 2.  \end{lemma} \begin{proof} Consider a group size
	$\gsizevec$ (stripped of the superscript $r$ for notation
	convenience) derived by a dataset $D$ and let $\gsizevec'$ be
	the group size generated by a neighboring dataset $D'$ of
	$D$. Let $D' = D \cup \{(\user, \unit, \region, 1)\}$ be the
	dataset that adds an individual to some group $G_u$, for some
	$u \in \unitset$, associated with a group size $\gsize_i$.
	This action decreases the group size $\gsize_i$ value by one
	and increases the group size $\gsize_{i+1}$ value by one,
	i.e., $\gsize_i' = \gsize_i + 1$ and $\gsize_{i+1}'
	= \gsize_{i+1} - 1$.  Therefore $\|\gsizevec' - \gsizevec\|_1
	= 2$.  A similar argument applies to a neighboring dataset
	that removes an individual from $D$.
\end{proof}

\begin{theorem}
\label{th:privacy_H}
$\cM_H$ satisfies $\epsilon$-differential privacy.
\end{theorem}

\begin{proof}
First note that each level of the hierarchy forms a partition over the
regions in $\regionset$ (by definition of hierarchy).  By parallel
composition (\Cref{th:par_composition}), for each level $\ell \in
[L]$, the noisy values $\tilde{\gsizevec}^r$ for all $r \in
\regionset_\ell$ satisfy $\frac{\epsilon}{L}$-differential privacy.  There are
exactly $L$ levels in the hierarchy: Therefore, by sequential
composition (\Cref{th:seq_composition}), the mechanism satisfies
$\epsilon$-differential privacy.
\end{proof}

\begin{theorem} 
\label{th:treedp_compl}
Constructing $\hat{\cT}^{\textsl{dp}}$ requires
solving $O(|\regionset|\,N\,\bar{D}$) optimization problems given
in \Cref{eq:d}, where $\bar{D} = \max_{s,r} |D_s^r|$ for
$r\in \regionset, s\in [N]$.  
\end{theorem}

\begin{proof}
	There are exactly $|\regionset|N$ nodes, and each node runs the program in \Cref{eq:d} for each element of its domain $\dom_s^r$. 
\end{proof}

The next lemma is the key technical result of the paper and
it requires some additional notation. Let $(v_0, \ldots, v_n)$ be the
ordered sequence of values in $\dom^r$, i.e., $v_{i+1} = v_i +1$ for
all $i=0 \ldots n-1$. Given a cost table $\ctable^r$, its associated
\emph{cost function} $\cfunction^r$ is a PWL function defined for all
$x \in [v_0, v_n]$ as:

\begin{equation}
	\cfunction^r(x) = 
	\left\{
		\begin{array}{l l}
		(x - v_0) (\ctable^r(v_1) - \ctable^r(v_0)) + \ctable^r(v_0) & \text{if } v_0 \leq x < v_1 \\
		(x - v_1) (\ctable^r(v_2) - \ctable^r(v_1)) + \ctable^r(v_1) & \text{if } v_1 \leq x < v_2 \\
		\ldots & \\
		(x - v_{n-1}) (\ctable^r(v_n) - \ctable^r(v_{n-1})) + \ctable^r(v_{n-1}) & \text{if } v_{n-1} \leq x \leq v_n.
		\end{array}
	\right.
\end{equation}
Similarly, $\cfunctionch^r$ is the PWL function defined from $\ctablech^r$ by the same process. 

\begin{lemma}
\label{lemma:phi_pcl}
Consider a given region $r$ and group size $s$. If $\cfunction_s^c$ is CPWL for all $c \in ch(r)$, then $\cfunctionch^r_s$ is CPWL.
\end{lemma}

\begin{proof} For simplicity, the proof omits subscript $s$ from
$\cfunction^r_s$, $\ctable^r_s$, $\ctablech^r_s$ and
$\cfunctionch^r_s$. The proof also uses $\sum_c$ and $\min_c$ to
denote $\sum_{c \in ch(r)}$ and $\min_{c\in ch(r)}$, respectively.

The proof is by induction on the levels in $\tree$. For the base case,
consider any leaf node $a^r$. Its cost table $\ctable^r$ only uses
\Cref{eq:d1} which is CPWL. Hence $\cfunction^r$ associated with
$\ctable^r$ is CPWL. The remainder of the proof shows that, if the
statement holds for $l+1,\ldots,L$, it also holds for level
$l$.

Consider a node $\node^r$ at level $l$. By induction, the cost
function $\cfunction^c$ ($c \in ch(r)$) is CPWL. Now consider for
each node $c \in \textsl{ch}(r)$ a value $v_c^0$ with minimum cost,
i.e., $v^0_c = \argmin_v \ctable^c(v)$. As a result, the value $V^0 =
\sum_{c} v^0_c$ has minimal cost $\ctablech^r(V^0) = \sum_c
\ctable^c(v^0_c)$.  Having constructed the minimum value in cost table
$\ctablech^r$, it remains to compute the costs of all values $V^0+k$
for all integer $k \in [1, \max D^r - V^0]$ and of all values $V^0-k$
for all integer $k \in [1, V^0 - \min D^r]$. The proof focuses
on the values $V^0+k$ since the two cases are similar. 

Recall that the algorithm builds a sequence of vectors
$\bv^0,\bv^1,\ldots,\bv^k,\ldots$ that provides the optimal
combinations of values for
$\ctablech^r(V^0),\ctablech^r(V^0+1),\ldots,\ctablech^r(V^0+k),\ldots$.
Vector $\bv^k$ is obtained from $\bv^{k-1}$ by changing the value of a
single child whose cost function has the smallest slope, i.e.,
\begin{align}
\label{eq:special_v_p}
v_c^{k} = \left\{ 
\begin{array}{l l}
v_c^{k-1} + 1 & \mbox{if }\; c = \argmin_c 
		\ctable^c(v_c^{k-1} + 1) - \ctable^c(v_c^{k-1}) \\ 
v_c^{k-1}     & \mbox{otherwise} 
\end{array}\right.
\end{align}
Observe that, by construction, the vector
$\bv^0,\bv^1,\ldots,\bv^k,\ldots$ satisfy the consistency constraints
\eqref{eq:d3} for
$\ctablech^r(V^0),\ctablech^r(V^0+1),\ldots,\ctablech^r(V^0+k),\ldots$,
since only one element is added at each step.

The next part of the proof is by induction on $k$ and shows that $\bv^k$ is the optimal
solution for $\ctablech^r(V^0 + k)$ for $k \geq 0$. The base case
$k=0$ follows by construction. Assume that $\bv^k$ is the optimal
solution of $\ctable^r(V^0+k)$. The proof shows that $\bv^{k+1}$ is
optimal for $\ctable^r(V^0+k+1)$. Without loss of generality, assume
that the first child ($c_1$) is selected in \eqref{eq:special_v_p}. It
follows that for $c \neq c_1$,
\begin{align*}
\ctable^{c_1}(v_{c_1}^k+1) - \ctable^{c_1}(v_{c_1}^k) \leq \ctable^c(v_c^k+1) - \ctable^c(v_c^k).
\end{align*}
The proof is by contradiction and assumes that there exists an optimal
solution $\bv^*$ such that $\sum_c \bv_c^* = V^0 + k + 1$ and
$\ctablech^r(\bv^*) < \ctablech^r(\bv^k)$. 
There are three cases to consider:

\begin{enumerate}[topsep=1ex,itemsep=-1ex,partopsep=1ex,parsep=1ex]

\item $v_{c_1}^* = v_{c_1}^{k}+1$: Let $\bv^+$ be the vector defined by $v_{c_1}^+
= v_{c_1}^*-1, v_c^+ = v_c^* \; (c \neq c_1)$.  $\bv^+$ satisfies the
consistency constraint for $V^0 + k$. Since $\ctablech^r(\bv^*) < \ctablech^r(\bv^k)$,
by hypothesis
and $v_{c_1}^* = v_{c_1}^{k}+1$, it follows that 
$\sum_c \ctable^c(v^*_c) < \sum_c \ctable^c(v^k_c)$.
But since $v_{c_1}^+ = v_{c_1}^k$, it follows that $\ctablech^r(\bv^+) < \ctablech^r(\bv^k)$ which contradicts the optimality of $\bv^k$.

\item $v_{c_1}^{*} > v_{c_1}^{k}+1$: By optimality of $\bv^k$ and $\bv^*$, the following properties hold:
\begin{align}
\ctable^{c_1}(v_{c_1}^{k}) + \sum_{c\neq c_1} \ctable^{c}(v_c^{k}) \leq \ctable^{c_1}(v_{c_1}^{*}-1) + \sum_{c\neq c_1} \ctable^{c}(v_c^{*}) \label{eq:P1}\tag{P1}\\
\ctable^{c_1}(v_{c_1}^{k}+1) + \sum_{c\neq c_1} \ctable^{c}(v_c^{k}) > \ctable^{c_1}(v_{c_1}^{*}) + \sum_{c\neq c_1} \ctable^{c}(v_c^{*}) \label{eq:P2}\tag{P2}
\end{align}
Subtracting \eqref{eq:P1} from \eqref{eq:P2} gives
\begin{align*}
\ctable^{c_1}(v_{c_1}^{k}+1) - \ctable^{c_1}(v_{c_1}^{k}) > \ctable^{c_1}(v_{c_1}^{*}) - \ctable^{c_1}(v_{c_1}^{*}-1)
\end{align*}
which is impossible since $v_{c_1}^{*} > v_{c_1}^{k}+1$ and $\cfunction^{c_1}$ is CPWL. 

\item $v_{c_1}^{*} < v_{c_1}^{k}+1$: The optimality of $\bv^*$ implies that
\begin{align}
\ctable^{c_1}(v_{c_1}^{*}) + \sum_{c \neq c_1} \ctable^{c}(v_c^{*}) < \ctable^{c_1}(v_{c_1}^{k}+1) + \sum_{c\neq c_1} \ctable^{c}(v_c^{k}) \label{eq:P3}\tag{P3}
\end{align}
Since $v_{c_1}^{*} < v_{c_1}^{k}+1$, there exists a child $c_2$ such that $v_{c_2}^* > v_{2}^k$. \eqref{eq:P3} can be rewritten as
\begin{align}
\ctable^{c_1}(v_{c_1}^{*})\!+\!\ctable^{c_2}(v_{c_2}^{*})\!+\!\sum_{c \neq c_1, c_2} \ctable^{c}(v_c^{*}) < \ctable^{c_1}(v_{c_1}^{k}+1)\!+\!\ctable^{c_2}(v_{c_2}^{k})\!+\!\sum_{c\neq c_1,c_2} \ctable^{c}(v_c^{k}). \label{eq:P4}\tag{P4}
\end{align}
The optimality of $\bv^k$ implies that
\begin{align}
\ctable^{c_1}(v_{c_1}^{*})\!+\!\ctable^{c_2}(v_{c_2}^{*}-1)\!+\!\sum_{c \neq c_1, c_2} \ctable^{c}(v_c^{*}) \geq \ctable^{c_1}(v_{c_1}^{k})\!+\!\ctable^{c_2}(v_{c_2}^{k})\!+\!\sum_{c\neq c_1,c_2} \ctable^{c}(v_c^{k}). \label{eq:P5}\tag{P5}
\end{align}
Substracting \eqref{eq:P5} from \eqref{eq:P4} gives
\begin{align*}
\ctable^{c_2}(v_{c_2}^{*}) - \ctable^{c_2}(v_{c_2}^{*}-1) < \ctable^{c_1}(v_{c_1}^{k}+1) - \ctable^{c_1}(v_{c_1}^{k}).
\end{align*}
Since $v_{c_2}^* > v_{c_2}^k$ and $\cfunction^{c_2}$ is CPWL, it follows that
\begin{align*}
\ctable^{c_2}(v_{c_2}^{k} + 1) - \ctable^{c_2}(v_{c_2}^{k}) < \ctable^{c_1}(v_{c_1}^{*}+1) - \ctable^{c_1}(v_{c_1}^{*})
\end{align*}
which contradicts the selection of $c_1$ in the algorithm. 
\end{enumerate}

\noindent
It remains to show that the resulting function $\cfunctionch^r$ is
convex, i.e.,
\begin{align}
\forall k \geq 0: \cfunctionch^r(V^0 + k + 2) - \cfunctionch^r(V^0 + k + 1) \geq \cfunctionch^r(V^0 + k + 1) - \cfunctionch^r(V^0 + k)
\end{align}
Let $c_k$ be the child selected when computing $\cfunctionch^r(V^0 + k)$ in \Cref{eq:special_v_p}. By selection of $c_k$, 
\begin{align*}
\forall c: \ctable^{c_k}(v_{c_k}^{k}+1) - \ctable^{c_k}(v_{c_k}^{k}) \leq \ctable^{c}(v_{c}^{k}+1) - \ctable^{c}(v_{c}^{k}).
\end{align*}
Since $\cfunction^{c}$ is CPWL, it follows that
\begin{align*}
\ctable^{c_k}(v_{c_k}^{k}+1) - \ctable^{c_k}(v_{c_k}^{k}) \leq \ctable^{c}(v_{c}^{k}+i+1) - \ctable^{c}(v_{c}^{k}+i)
\end{align*}
for all $i \geq 0$. The results follows by induction. 
\end{proof}

\begin{theorem}
\label{tm:convexity-a}
The cost table $\ctable_s^r$ is CPWL.
\end{theorem}
\begin{proof}
By \Cref{lemma:phi_pcl}, $\ctablech_s^r$ is CPWL. The function $d(v) \defeq (v - \tilde{n}^r_s)^2$ also defines a CPWL. The result follows from the fact that the sum of two
CPWL functions is CPWL. 
\end{proof}

% Algorithm \ref{alg:merge} runs in $O \big(\bar{D}_i^r \log \bar{D}_i^r \!+\!\bar{D}_i^r \big)$, where $\bar{D}_i^r = \sum_{s \in \textsl{ch}(r)} |D_i^s|$. This is a substantial speedup from the 
% $O\big(|D_i^r| \times (2\delta+1)^{|\textsl{ch}(r)|} \big)$ required by the algorithm presented in the previous section. 
% Our experimental result will also validate such claim (see Section \ref{sec:experiments}).

\begin{theorem}
	\Cref{alg:merge} requires requires $O\big(\bar{D} \log \bar{D} \big)$ operations, where $\bar{D} = \max_{s,r} \bar{D}_s^r$ for $r\in \regionset, s\in [N]$.
\end{theorem}
The result above can be derived observing that the runtime complexity of \Cref{alg:merge} is dominated by the sorting operations in lines $1$ and $2$.

\begin{lemma}
The sensitivity $\Delta_{\csizevec}$ of the cumulative group estimate query is 1.
\end{lemma}

\begin{proof}
Consider a vector of cumulative group sizes $\bm{c} = (c_1, \ldots, c_n)$. Additional, let $\bm{c}'$ be a vector of cumulative group sizes that differs from $\bm{c}$ by adding one individual to a group of size $k \in [n]$. It follows that $c_k' = c_k - 1$ but none of the other groups $c_i'$ changes: i.e., $c_i' = c_i$, for all $i \neq k$. 
Similarly, removing one individual from a group of size $k$ implies that $c_{k+1}' = c_{k+1} + 1$ and $c_{i}' = c_i$, for all $i \neq k$.

Therefore, for maximal change between any two neighboring datasets $\bm{c}$ and $\bm{c}'$ is $1$.
\end{proof}

% -- Changes start -- %
\begin{lemma}
The cost table $\ctable_i$ of each node $\node_i$ of $\chaindp$ is CPWL.
\end{lemma}

\begin{proof}
    Like \Cref{lemma:phi_pcl}, we prove this argument by induction. For the base case,
    consider the head node $\node_1$ of $\chaindp$. Its cost table is simply given by
    $\ctable_1(v)=\left(v-\tilde{\csize}_1\right)^2$ for any $v\in D_1$ and proves to be
    CPWL due to convexity of L2-Norm.
    
    Suppose that the cost table $\ctable_i$ is CPWL. If the nodes $\node_i$ and $\node_{i+1}$ are at the same level, 
    there exists an inequality constraint between the post-processed values of
    these two adjoining nodes, $\node_i$ and $\node_{i+1}$. Thus, the function $\ctablech_{i+1}(v)$
    is given in the following formula.
    \begin{equation*}
        \ctablech_{i+1}(v)=\min_{ \substack{x_i\in \dom_i\\ x_i\leq v }}\ctable_i(x_i)=
        \left\{
            \begin{array}{l  l}
                \ctable_i(v_i^0) & \text{if } v\geq v_i^0 \\
                \ctable_i(v) & \mathrm{otherwise},
            \end{array}
        \right.
    \end{equation*}
    where the value $v_i^0$ represents the minimizer of the cost table $\ctable_i$, i.e., 
    $v_i^0\coloneqq \argmin_{v\in \dom_i}\ctable_i(v)$. If follows that the function 
    $\ctablech_{i+1}(v)$ is CPWL. In the other case, an equality constraint is between the two post-processed
    values, which indicates that the function $\ctablech_{i+1}(v)$ is simply the cost table 
    $\ctable_i(v)$ and, thus, CPWL. Therefore, for both cases, the function $\ctablech_{i+1}(v)$ 
    is shown to be CPWL. Recall that 
    the function $d(v)\defeq (v - \tilde{c}_{i+1})^2$ also defines a CPWL. It follows that
    the cost table $\ctable_{i+1}$ enjoys the property of CPWL as well due to the fact that 
    the sum of two CPWL functions is CPWL.
\end{proof}

\begin{theorem}
The cost table $\ctable_i$ for each $i\in[\vert\regionset\vert N]$ can be computed 
in time $O\big(\bar{D} \big)$, where $\bar{D} = \max_{i} |D_i|$ for
$i\in[\vert\regionset\vert N]$.
\end{theorem}

\begin{proof}
    Consider the head node $\node_1$ of the hierarchy $\chaindp$. It takes $O\big(\bar{D} \big)$ operations to
    compute the cost table $\ctable_1$ and identify its minimizer $v_1^0$ via linear search. Given the cost
    table $\ctable_i$ and its associated minimizer $v_i^0$, we are able to generate the function $\ctablech_{i+1}(v)$
    in $O\big(\bar{D} \big)$, regardless of the type of the constraint between the two adjoining
    nodes, $\node_i$ and $\node_{i+1}$. Thus, we can update the cost table $\ctable_{i+1}$ and 
    and compute its minimizer $v_{i+1}^0$ in $O\big(\bar{D} \big)$, for any $i\in [\vert\regionset\vert N-1]$.
\end{proof}
% -- Changes end -- %

\end{document}